\documentclass{article}
\usepackage{style}
\usepackage{mdframed}
\usepackage{tikz}
\usepackage{float}
\usepackage{makecell}
\usepackage{enumitem}
\usepackage{xspace}

\usepackage{thmtools}

\usepackage{graphicx} %

\providecommand{\Comments}{1}
\ifnum\Comments=1
\usepackage[colorinlistoftodos,prependcaption,textsize=scriptsize]{todonotes}
\definecolor{Gred}{RGB}{219, 50, 54}
\definecolor{Ggreen}{RGB}{60, 186, 84}
\definecolor{Gblue}{RGB}{72, 133, 237}
\definecolor{Gyellow}{RGB}{204, 204, 0}
\definecolor{Gpurple}{RGB}{204, 0, 204}
\definecolor{Gorange}{RGB}{255, 200, 120}
\definecolor{Gbrown}{RGB}{0, 204, 204}
\paperwidth=\dimexpr \paperwidth %
\marginparwidth=\dimexpr\marginparwidth %
\else
\usepackage[disable]{todonotes}
\fi
\usepackage{colortbl}
\newcommand{\mytodo}[1]{\ifnum\Comments=1{#1}\fi}

\newcommand{\NP}{\textnormal{NP}}
\newcommand{\letterboxed}{\textnormal{LETTER BOXED}\xspace}

\usepackage{authblk}
\usepackage{hyperref}

\title{Man, these New York Times games are hard!\\
A computational perspective}

\author[*]{Alessandro~Giovanni~Alberti}
\author[*]{Flavio~Chierichetti}
\author[*]{Mirko~Giacchini}
\author[*]{Daniele~Muscillo}
\author[*]{Alessandro~Panconesi}
\author[*]{Erasmo~Tani}
\affil[*]{Sapienza University of Rome  

Emails: \texttt{\{alberti.2050512, muscillo.2080466\}@studenti.uniroma1.it}, \texttt{\{flavio, giacchini, ale, tani\}@di.uniroma1.it}
}
  
\date{\today}

\begin{document}
\maketitle

\begin{abstract}
    \noindent
    The New York Times (NYT) games have found widespread popularity in recent years and reportedly account for an increasing fraction of the newspaper's readership. In this paper, we bring the computational lens to the study of New York Times games and consider four of them not previously studied: Letter Boxed, Pips, Strands and Tiles. We show that these games can be just as hard as they are fun. In particular, we characterize the hardness of several variants of computational problems related to these popular puzzle games. For Letter Boxed, we show that deciding whether an instance is solvable is in general $\NP$-Complete, while in some parameter settings it can be done in polynomial time. Similarly, for Pips we prove that deciding whether a puzzle has a solution is $\NP$-Complete even in some restricted classes of instances. We then show that one natural computational problem arising from Strands is $\NP$-Complete in most parameter settings. Finally, we demonstrate that deciding whether a Tiles puzzle is solvable with a single, uninterrupted combo requires polynomial time.
\end{abstract}

\section{Introduction}
The New York Times is one of the most influential newspapers in the world.  In its 173 years of history, it has reported on events as disparate as the American Civil War, the sinking of the Titanic, the publication of Einstein's theory of relativity and the discovery of the tomb of Tutankhamun\footnote{It even announced the existence of life on Mars~\cite{W06} (could it be watermelons \cite{H21}?)}. In this time, the newspaper's reporting activities have won it over 130 Pulitzer prizes \cite{nyt_pulitzers}. The newspaper was also vastly successful in transitioning to the digital era, and its online newspaper boasts millions of subscribers~\cite{nyt2024intlsubs}, reportedly more than any other English language online newspaper in the world~\cite{m23}.

In spite of all the newspaper's remarkable achievements, the New York Times' website is increasingly visited for a different reason: a carefully curated collection of puzzle games \cite{Klein2023NYTimesGames, Knight2024AtlanticStrands, Switzer2024StrandsGamer}\footnote{In March of 2024, ValueAct Capital Management, a hedge fund investing in the New York Times Company, filed a notice of exempt solicitation under Rule 14a-6(g) \cite{WinNT}, in which they reported finding that in 2023 people spent more time playing NYT games than reading NYT news articles.}. In February 1942, in a move said to have been caused by the attack on Pearl Harbor~\cite{Shepard92}, the newspaper launched its famous crossword, which has since entertained millions of enthusiasts and become a household name in the crossword world. It is perhaps not surprising then, that the New York Times has leveraged its reputation as a one-stop shop for intellectuals and clever thinkers to break into the online puzzle game space.

\begin{figure}
    \centering
    \includegraphics[scale=0.7]{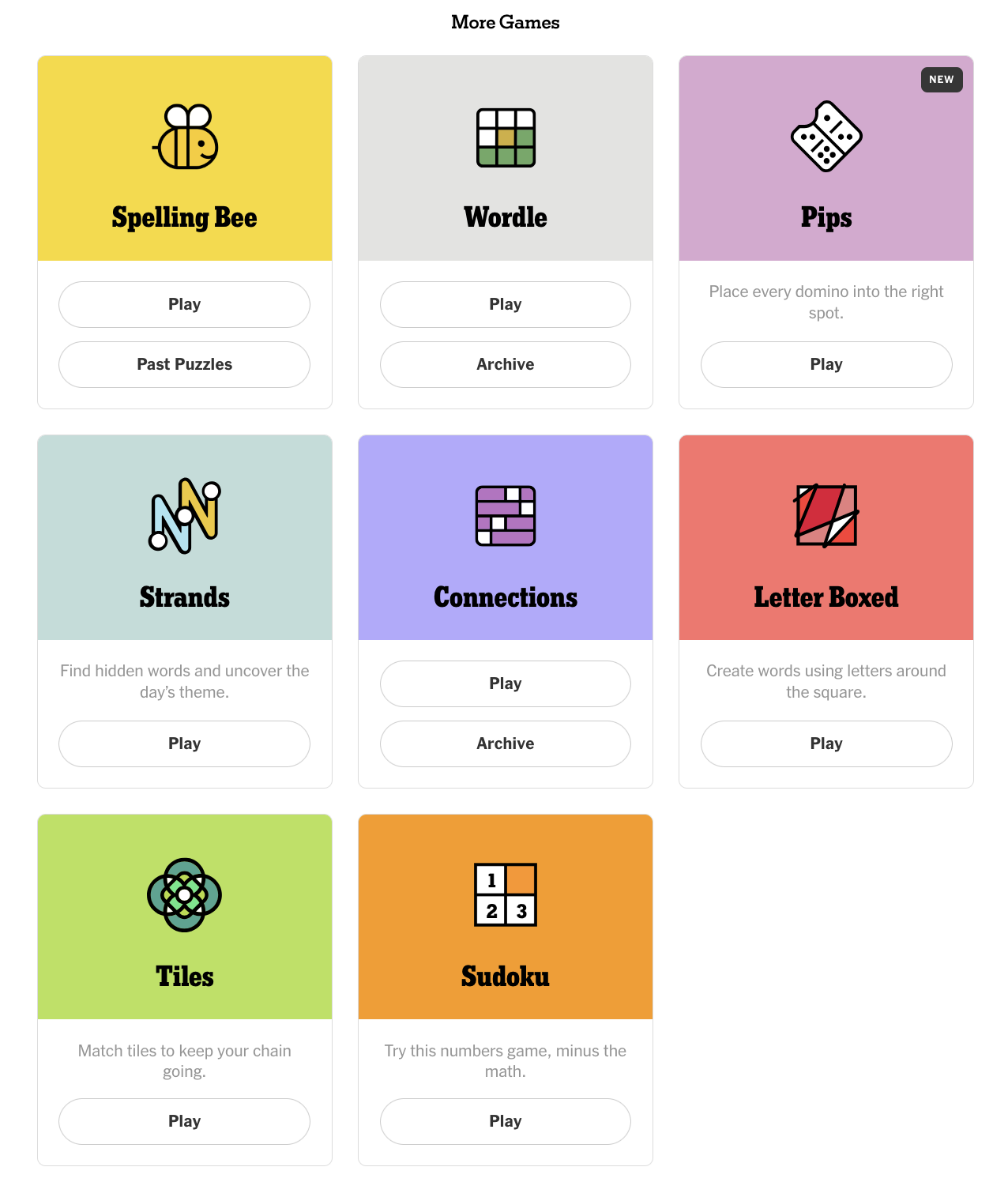}
    \caption{The ``More Games'' section of the \url{https://www.nytimes.com/crosswords} page. Accessed on 08/22/25. Used for research and illustrative purposes only.}
    \label{fig:nyt-games-homepage}
\end{figure}

Started in 2014 with the introduction of the Mini Crossword and later the word game \emph{Spelling Bee}, the games section of the New York Times has added more and more titles to its collection over the years. In 2019 they introduced \emph{Letter Boxed}, in which the players have to compose words using letters placed around a square, and \emph{Tiles}, in which one is asked to find tiles sharing visual features. In 2022, they also acquired the popular game \emph{Wordle}, which had already found a large following. In 2024 they added \emph{Strands} a word search-like game, and in 2025 they added the logic puzzle game Pips, in which a player has to position domino tiles on a board while satisfying given constraints.

According to a New York Times' spokesperson interviewed by The Verge, in 2024 alone the games were played over eleven billion times \cite{Peters2025NYTScrabble}.

Despite the popularity of the New York Times games, so far, with the exception of Sudoku~\cite{ys03}, Wordle~\cite{LS22wordle,r22} and the traditional crossword~\cite{ghlm24} the games have not been studied and understood from a computational standpoint, and their complexity has remained unknown.

In this paper, we attempt to right this unspeakable wrong, and provide a computational perspective on several problems arising from the New York Times games. Our focus will be on four of the games, which are computationally rich and yet so far unstudied: Letter Boxed, Pips, Strands and Tiles. We consider various NYT games-inspired computational problems and analyze their computational complexity.

\subsection{Overview of The New York Times Games}\label{sec:overview-of-games}
As mentioned in the introduction, the New York Times offers a collection of games beyond its traditional crossword puzzles (see \Cref{fig:nyt-games-homepage}). In this paper, we will focus on four games: Letter Boxed, Pips, Strands and Tiles. In this section, we briefly review their rules. 

\newpage
\begin{figure}[h!]
    \centering
    \includegraphics[width=1\linewidth]{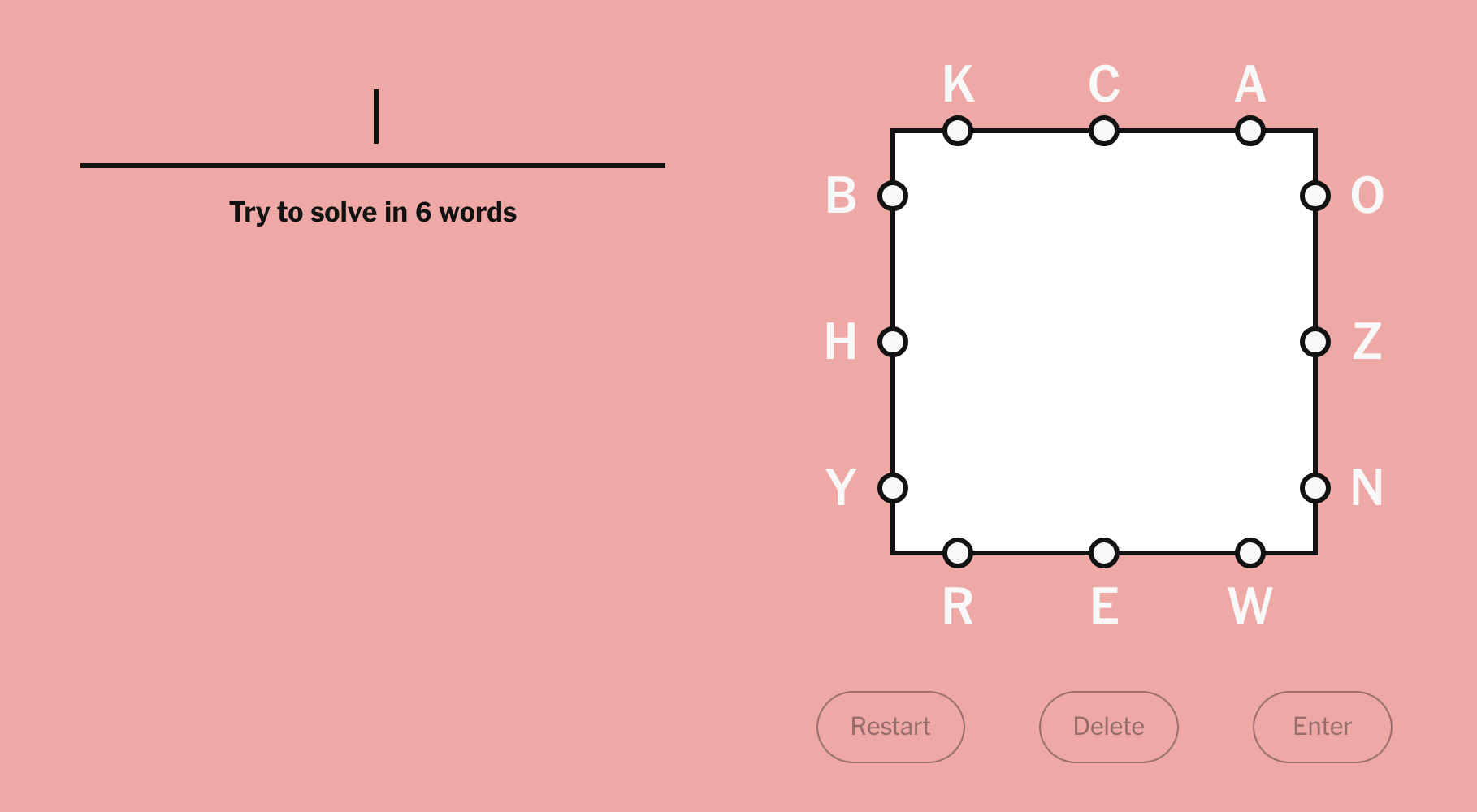}
    \caption{A typical instance of the Letter Boxed game. In this instance, the player is encouraged to complete the puzzle with at most $6$ words, which suggests that the puzzle may be harder than usual.\\ Accessed from \url{https://www.nytimes.com/puzzles/letter-boxed} on 06/26/25. Used for research and illustrative purposes only.}
    \label{fig:letter-boxed-game}
\end{figure}

\paragraph{Letter Boxed.} In Letter Boxed (\Cref{fig:letter-boxed-game}), the player is given a square with three letters on each of its sides. The letters are all distinct. The player can then start to compose English words using the letters on the sides of the square, with the following set of rules. First, the player may never use two letters from the same side consecutively in a word. Second, each word formed must start with the same letter (on the same side) that ended the previous word. Finally, every letter must appear in at least one of the player's words. When the player completes a word that includes all the letters not previously included in any word, the puzzle is solved.

The more skilled a player is, the fewer words they will need to complete a puzzle. Next to the puzzle, one is given a suggestion on how many words may be needed for an average player to complete it. This is typically given in the form of a number between $4$ and $6$.
\newpage
\begin{figure}[H]
    \centering
    \includegraphics[scale=0.4]{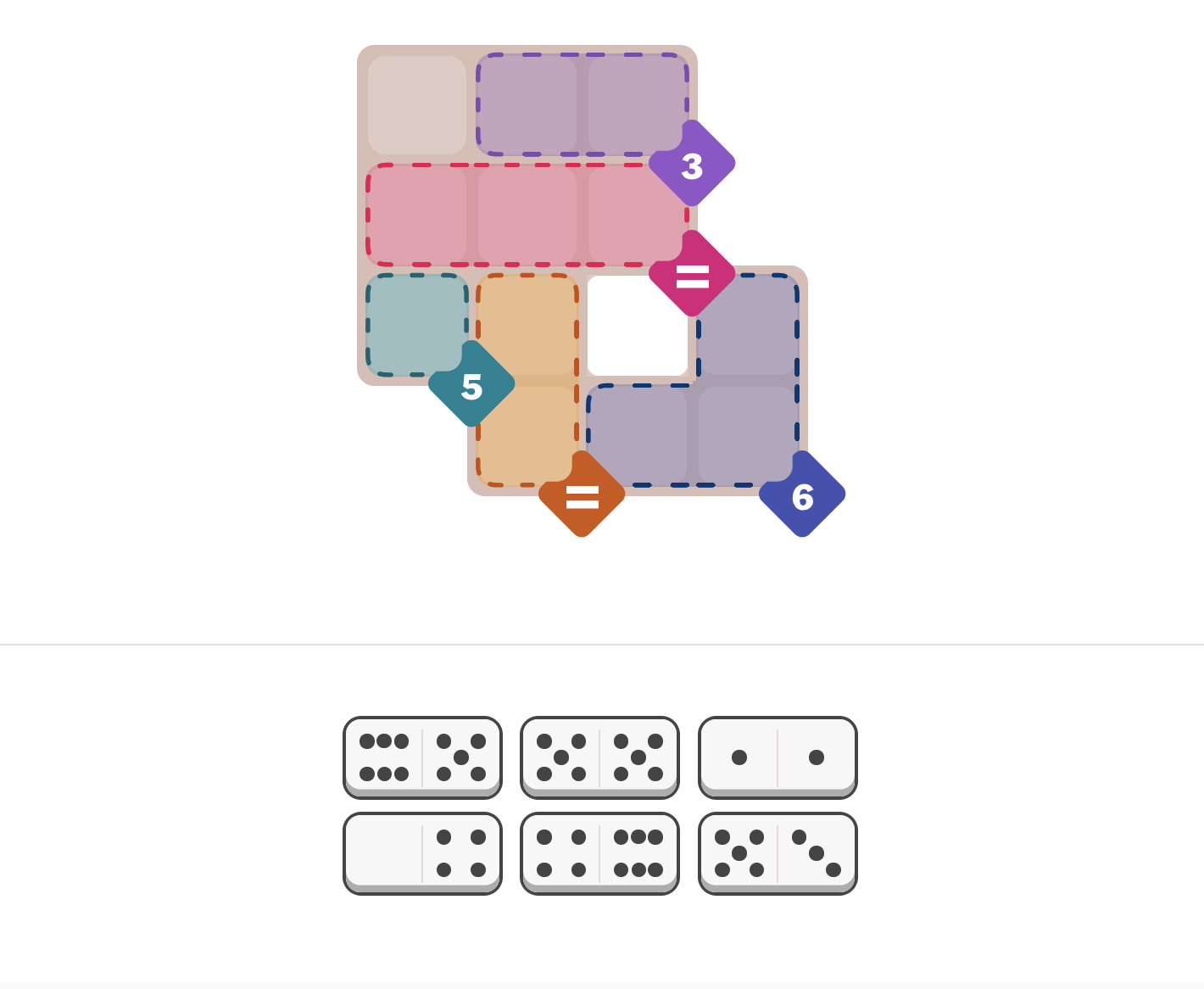}
    \caption{A medium difficulty instance of Pips. In this instance, the board contains five constraint regions. Three of them (labeled with the numbers 3, 5 and 6) specify a constraint on the sum of the domino squares placed on them, and two of them (labeled with the = symbol) specify that all the domino squares placed in the region must have the same value. The goal of the player is to tile the board with all the domino pieces displayed at the bottom, in a way that respects all the constraints. Each domino piece must be used exactly once in the solution. \\
    Accessed from \url{https://www.nytimes.com/games/pips/medium} on 08/22/2025. Used for research and illustrative purposes only.}
    \label{fig:pips-instance}
\end{figure}

\paragraph{Pips.} Pips (\Cref{fig:pips-instance}) is the newest entry in the NYT games family. In this game, introduced in the summer of 2025~\cite{nyt2025pips}, the player is given a collection of domino tiles and they have to place them on a board in such a way that each square of the board is covered by exactly one half of a domino tile (i.e.\ the tiles tessellate the board). The board also displays a collection of connected pairwise disjoint regions, each of which specifies a constraint on the values of the domino pieces placed within the region. There are three types of constraints. 
\begin{enumerate}
    \item ``='' (resp. ``$\neq$'') constraints specify that the numbers placed within the region must all be equal to (resp. distinct from) one another. 
    \item ``$n$'' constraints, for some number $n \in \mathbb{N}_0$, specify that the numbers within a specific region must sum to $n$. 
    \item ``<$n$'' (resp.\ ``>$n$'') constraints specify that the sum of the numbers inside the region must be smaller (resp.\ greater) than $n$. 
\end{enumerate}
The goal of the game is to place the domino tiles so as to fill the entire board in a way that satisfies all of the constraints specified.

\begin{figure}[H]
    \centering
    \includegraphics[width=1\linewidth]{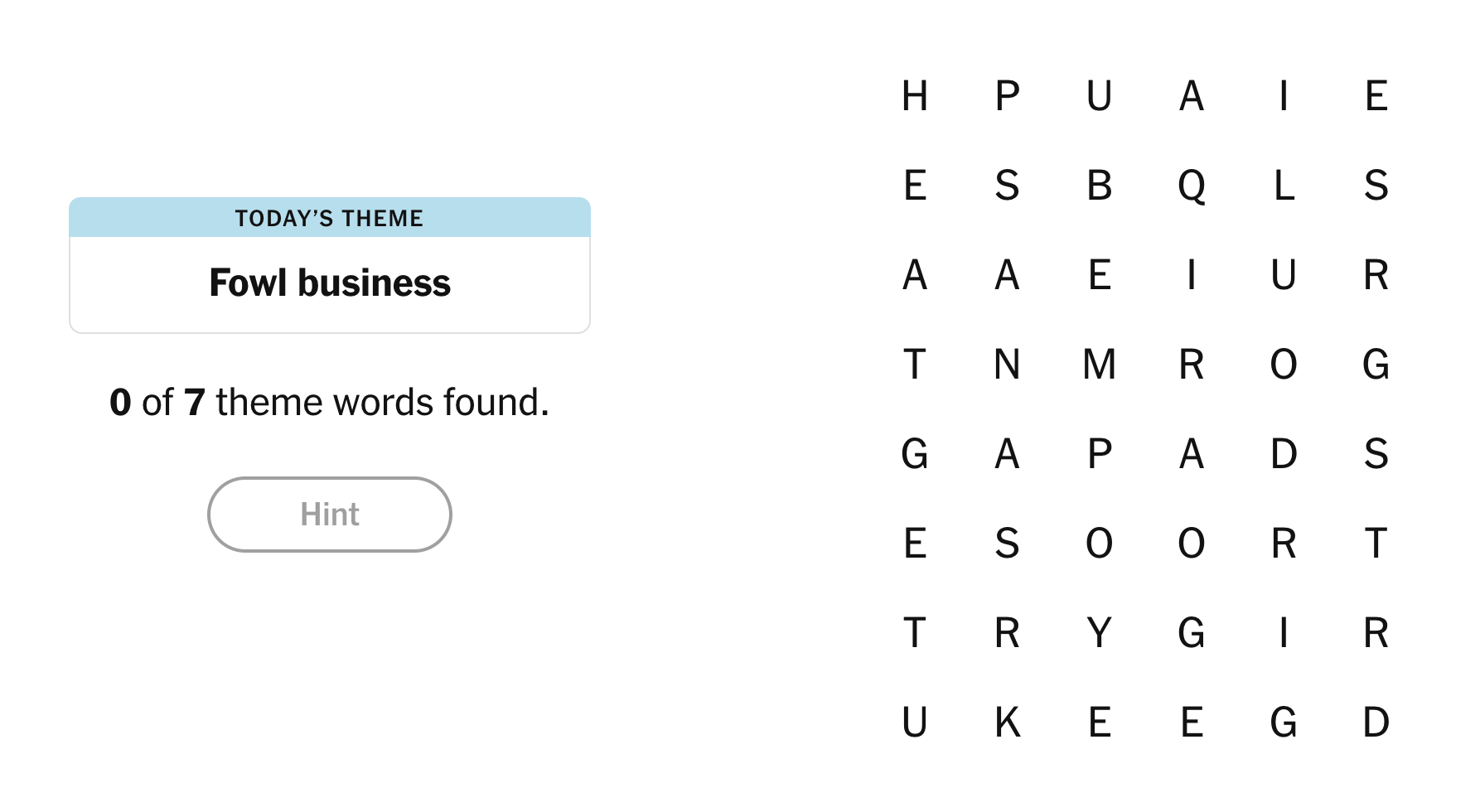}
    \caption{An instance of the Strands game. In this instance, the hint for the theme of the day is ``Fowl business''. This wordplay suggests that the theme of the day may be related to birds that are hunted. In fact, in the bottom-left corner one can easily find the word \emph{turkey}, and the spangram is the phrase \emph{game birds} which can be seen crossing the puzzle from left to right.\\
    Accessed from \url{https://www.nytimes.com/games/strands} on 06/26/2025. Used for research and illustrative purposes only.}
    \label{fig:strands-game}
\end{figure}

\paragraph{Strands.} Introduced in 2024 \cite{Levine2024Strands}, Strands (\Cref{fig:strands-game}) is a modern spin on classic word search puzzles. In this game, the player is given a matrix of letters, partitioned into a hidden collection of words (the solution to the puzzle). Each word in the solution is related to the theme of the day, hinted at in the box on the left. The goal is to find all the words in the solution. Each of these is formed by a sequence of letters, each adjacent to the next vertically, horizontally or diagonally (the same cell cannot be repeated). Moreover, each puzzle solution contains a \emph{spangram}, a word or phrase that connects two opposite sides of the matrix which explicitly reveals the theme of the day. 

When the player queries a word, it is revealed whether it belongs to the solution or not. After finding three valid words that do not belong to the hidden collection, the user is given a hint on one of the hidden words. The hint highlights the set of letters in one of the words in the solution. If one were to obtain another hint right after, this would also reveal the word itself, by showing the exact order in which the previously highlighted letters appear in the word.

\begin{figure}[H]
    \centering
    \includegraphics[width=0.8\linewidth]{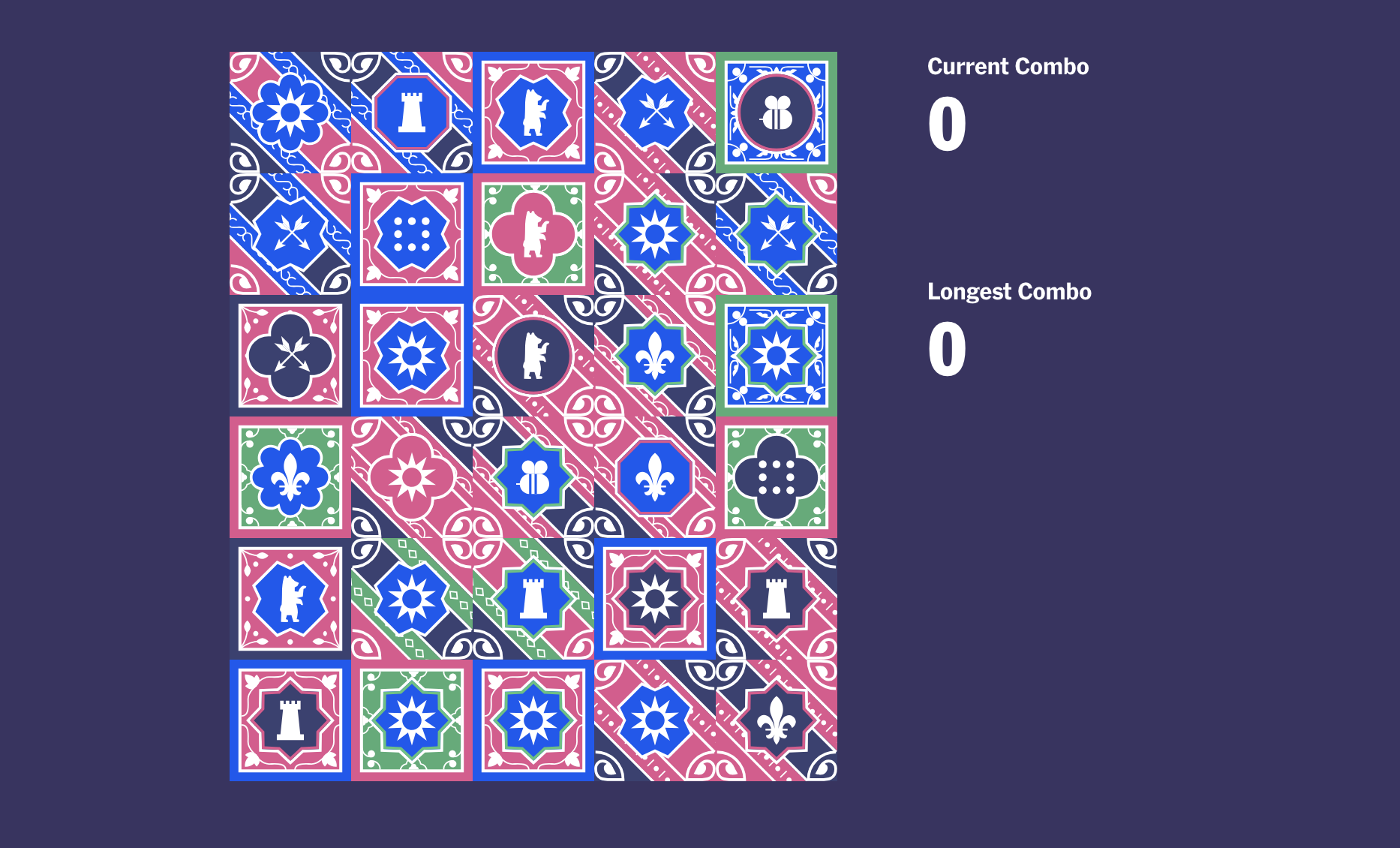}
    \caption{An instance of the Tiles game. \\
    Accessed from \url{https://www.nytimes.com/puzzles/tiles} on 06/26/2025. Used for research and illustrative purposes only.}
    \label{fig:tiles-game}
\end{figure}

\paragraph{Tiles.} In Tiles (\Cref{fig:tiles-game}), the player is presented with a grid of tiles. Each tile in the grid displays some visual elements (\emph{features}). The player begins the game by selecting an arbitrary tile. Each move consists of selecting a new tile to move to, distinct from the previous one. If the newly selected tile shares one or more non-deleted features with the previous one, then these features are deleted from both tiles, and the player's current combo increases by one. 
Otherwise, the player simply moves to the new tile. We refer to this latter case as a \emph{teleport move}. Note that, when a player finds themselves on a tile on which all the features have been deleted, they are forced to make a teleport move. Teleport moves do not extend the current combo. Moreover, if the player ever makes a teleport move that is not forced, then the length of their current combo is set to zero. The goal of the game is to delete all the features from  all the tiles.

\subsection{Our Results}\label{sec:our-results}
In this paper, we show several results on the hardness of New York Times puzzle games.

\paragraph{Our Results on Letter Boxed.} For our first set of results, in \Cref{sec:letter-boxed}, we consider the problem of deciding whether an instance of Letter Boxed is solvable with at most $k$ words. We generalize the game by allowing $n$ characters on each side of the square that we represent as explicit multisets, and we consider words coming from a dictionary $D$ with alphabet $\Sigma$. We give a complete characterization of the complexity of the game for different regimes of the alphabet size $|\Sigma|$, the dictionary size $|D|$, the length $L$ of the longest word in $D$, and the budget of words $k$.

\renewcommand{\arraystretch}{2}
\begin{theorem}\label{thm:letter-boxed-table}
    Consider the problem of deciding whether an instance of Letter Boxed is solvable with at most $k$ words. Then the complexity of the problem for different parameter regimes is as described in \Cref{tab:letterboxed-results}. %
    \begin{table}[h!]
    \begin{center}
    \begin{tabular}{|c c c c|c|}
    \hline
    \makecell{Alphabet Size \\ $|\Sigma|$} & \makecell{Dictionary Size \\ $|D|$} & \makecell{Maximum \\ Word Length\\ $L$} & \makecell{Budget of Words \\$k$ }& Complexity Class\\
    \hline
    Constant & * & * & * & P (\Cref{cor:letterboxed-poly-cases}) \\
    Arbitrary & Constant & Constant & * & P (\Cref{cor:letterboxed-poly-cases}) \\
    Arbitrary & Arbitrary & Constant & Constant & P (\Cref{cor:letterboxed-poly-cases}) \\
    Arbitrary & * & Arbitrary & * & {\NP}-Complete (\Cref{thm:letter-boxed-np-complete-d1-k1}) \\
    Arbitrary  & Arbitrary & Constant & Arbitrary & {\NP}-Complete (\Cref{thm:letter-boxed-np-complete})\\
    \hline
    \end{tabular}
    \end{center}
    \caption{The complexity of Letter Boxed. Asterisks (*) indicate that the result holds both if the parameter is constant and if it can grow with the input size. The settings solvable in polynomial time where parameter $k$ is marked as ``arbitrary'' hold even if $k$ is represented in binary notation, while the $\NP$-Hardness results hold even when $k$ is represented in unary notation (i.e.\ they are \emph{strong $\NP$-Hardness} results). The number of characters of each side, $n$, is always allowed to grow arbitrarily.}
    \label{tab:letterboxed-results}
    \end{table}
\end{theorem}

\paragraph{Our Results on Pips.} For Pips, we study the problem of deciding whether an instance of the puzzle admits a solution (\Cref{sec:pips}). We show that this problem is $\NP$-Complete, even in the special case in which all the domino tiles' faces have value either zero or one, and there are no ``$\neq$'', ``<$n$'' or ``>$n$'' constraints. We summarize the result in the following theorem.

\begin{theorem}\label{thm:pips-overview}
    The problem of deciding whether a Pips puzzle is solvable is $\NP$-Complete. This is true even in the restricted class of instances in which all the domino tiles' numerical values are either $0$ or $1$ and none of the constraints are of the form ``$\neq$'', ``<$n$'' or ``>$n$''.
\end{theorem}

Moreover, we show that if instances can have arbitrarily large numerical values on the domino tiles, it is weakly $\NP$-Hard to decide whether an instance with a single constraint can be solved. 

\paragraph{Our Results on Strands.} In \Cref{sec:strands} we consider the problem of deciding whether a Strands grid $M\in \Sigma^{n\times m}$ can be partitioned into valid words of a dictionary $D$, where $\Sigma$ is an alphabet of characters. Ignoring the \emph{spangram} that can always be added in the last row, we show that this problem is $\NP$-Hard even if the instance contains only constantly many words of constant length. 

\begin{restatable}{theorem}{StrandsNPComplete}\label{thm:strands-np-complete}
Given a Strands instance $(\Sigma, D, M)$ determining whether $M$ can be partitioned into valid words from $D$ is $\NP$-Complete, even if $|\Sigma|, |D|, L = O(1)$, where $L$ denotes the maximum length of a word in $D$. 
\end{restatable}

\paragraph{Our Results on Tiles.} In \Cref{sec:tiles}, we consider the problem of deciding whether a Tiles puzzle can be completed. We show that this condition is equivalent to each distinct feature appearing in an even number of tiles, and hence the problem can be solved in linear time. This result is summarized in the following theorem.

\begin{restatable}{theorem}{TilesMainTheorem}\label{thm:tiles-main-result}
    Given any instance of Tiles, the following are all equivalent:
    \begin{enumerate}
        \item The instance is solvable,
        \item The instance is solvable with a single, unbroken combo,
        \item Each distinct feature in the instance appears in an even number of tiles.
    \end{enumerate}
    In particular, there exists an algorithm that decides whether an instance of Tiles is solvable in time linear in the size of the instance.
\end{restatable}
Here, the size of the puzzle is the number of features appearing in every tiles (counted with multiplicity).

We also consider the problem of deciding whether a Tiles puzzle can be completed without making use of teleport moves. We show that this problem is in P for a restricted class of puzzles: those in which any two tiles share at most one feature. The complexity of this problem in the general setting is left as an open problem.

\section{Letter Boxed}\label{sec:letter-boxed}
\vspace{-14mm}
\begin{figure}[H]
    \centering    \includegraphics[width=1\linewidth]{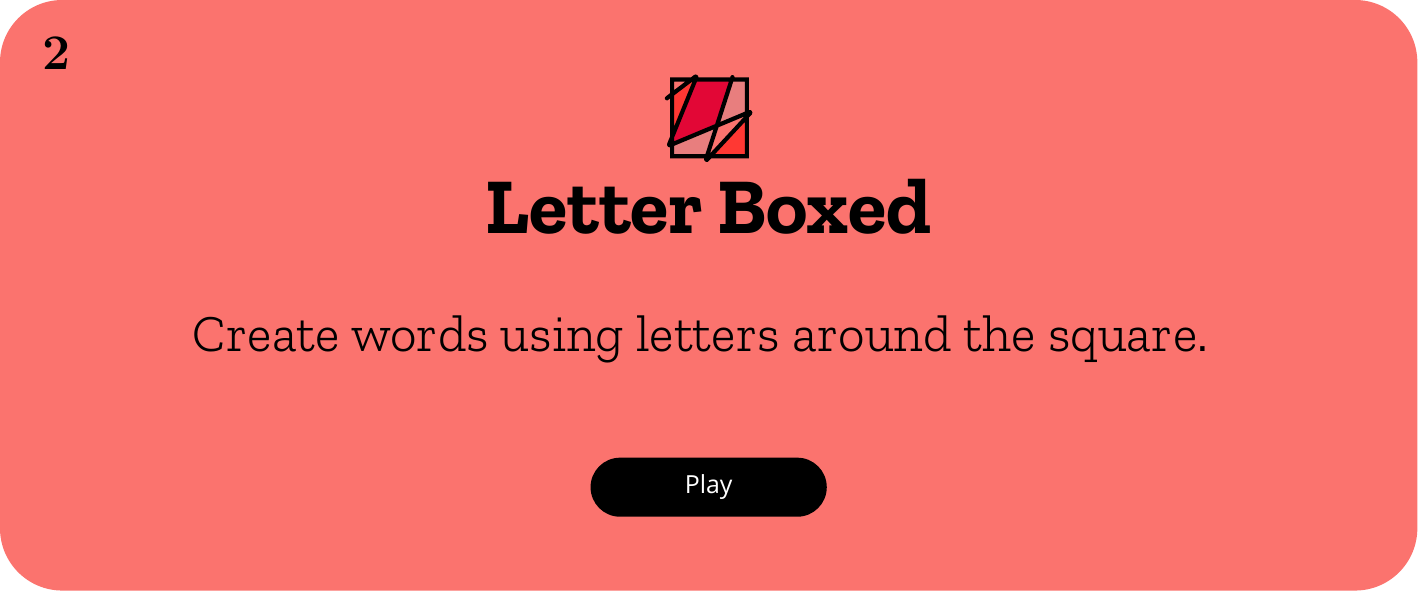}
\end{figure}

In this section, we discuss the complexity of some computational problems related to Letter Boxed. The central problem of interest will be determining the minimum number of words needed to complete a given Letter Boxed puzzle. We will show that this problem is $\NP$-Hard in general, but becomes tractable for certain choices of parameters. The results in this section are summarized in \Cref{thm:letter-boxed-table} in \Cref{sec:our-results}.

Recall that a typical instance of Letter Boxed is given by a square, with three letters on each side. In order to obtain a parametrized problem, we allow our instances to have $n$ characters on each side for a varying parameter $n$. Since we allow the set of distinct characters (alphabet) to be smaller than $4n$, we study a more general version of the problem in which characters could be repeated, even on the same side of the square.

Note that the specific position of each character on a side is immaterial to the game: the game is entirely isomorphic up to permuting characters on the same side of the square. We therefore define instances as being specified up to this isomorphism.%

\begin{definition}[Letter Boxed puzzle]
    A \emph{Letter Boxed puzzle} is a tuple $(\Sigma,D,\Gamma_1,\Gamma_2,\Gamma_3,\Gamma_4)$ where $\Sigma$ is a finite alphabet, $D \subseteq \Sigma^*$ is a finite set of strings of characters from $\Sigma$ called the \emph{dictionary}, and each $\Gamma_i$ is a multiset of $n$ elements (counted with multiplicity) from $\Sigma$, representing the characters appearing on the $i^{th}$ side of the square.
\end{definition}
Given a Letter Boxed puzzle, we denote by $L$ the length of the longest word in $D$. Note that we explicitly represent the multisets $\Gamma_1, \dots, \Gamma_4$ so that the size of the puzzle is $|\Sigma|+4n + \sum_{w\in D}|w|$. 

In order to solve the problem, one needs to find a sequence of words that covers all the characters on every side, and such that each word in the sequence ends with the same character with which the next word starts. Moreover, these words need to satisfy the constraint that no two subsequent characters are on the same side. We formalize these conditions in the following definition.
\begin{definition}[Letter Boxed Solution]\label{def:solution-to-letter-boxed}
    Given a Letter Boxed puzzle $(\Sigma,D,\Gamma_1,\Gamma_2,\Gamma_3,\Gamma_4)$, a \emph{solution of length $k$} to the puzzle is a pair $(\texttt{words}, \texttt{sides})$, where $\texttt{words} = (w^{(1)}, \dots , w^{(k)}) \in D^{k}$ is an ordered sequence of $k$ words from the dictionary, and $\texttt{sides} \in [4]^{\sum_{i\in [k]} |w^{(i)}|}$ specifies a side for every character in every word in \texttt{words}, with the following properties:
    \begin{enumerate}
        \item For every $i =1, \dots , k-1$, the last character of $w^{(i)}$ is the same as the first character of $w^{(i+1)}$.
    \end{enumerate}
    Moreover, letting $\sigma = w^{(1)}\circ w^{(2)} \circ \cdots \circ w^{(k)}$ be the concatenation of the words in $\texttt{words}$ then:
    \begin{enumerate}
        \item[2.] For every $i$, $\texttt{sides}_i \neq \texttt{sides}_{i+1}$, unless $\sigma_i$ is the end of some word $w^{(j)}$, in which case $\texttt{sides}_i = \texttt{sides}_{i+1}$,
        \item[3.] For every $i\in[|\sigma|]$, $\sigma_i \in \Gamma_{\texttt{sides}_i}$,
        \item[4.] For every $i\in[4]$, and for every character $\gamma \in \Gamma_i$ the number of indexes $j$ such that: (i) $\sigma_j = \gamma$, (ii) $\texttt{sides}_j=i$,  and (iii) either $j=|\sigma|$ or $\sigma_j$ is not the end of some $w_\ell$, is at least the number of occurrences of $\gamma$ in $\Gamma_i$. Note that the third property ensures that we do not double count characters when starting a new word.
    \end{enumerate}
    If a Letter Boxed puzzle admits a solution of length $k$ then we say it is \emph{solvable with at most $k$ words}.
\end{definition}

We consider the following decision problem.
\begin{mdframed}
    {LETTER BOXED}\\
    \textit{Input:} A Letter Boxed puzzle $(\Sigma, D, \Gamma_1, \Gamma_2, \Gamma_3, \Gamma_4)$ and a positive integer $k$. \\
    \textit{Question:} Can the Letter Boxed puzzle be solved with at most $k$ words?
\end{mdframed}

In this work we will also consider variants of the game with a number of sides different from four. Here, we think of instances as being formed by putting letters on the sides of a regular $S$-polygon. Instances of this version of the game are defined by extending the decision problem \letterboxed (and the corresponding Letter Boxed puzzle) in the natural way. We call this version $S$-SIDES \letterboxed. Note that under this definition \letterboxed is simply $4$-SIDES \letterboxed.

The rest of this section proves \Cref{thm:letter-boxed-table}. We will do so via a series of results. We start by providing a natural dynamic programming algorithm for the problem, which will be used to prove the cases solvable in polynomial time. 

\begin{theorem}\label{thm:letterboxed-main-algorithm}
    There is an algorithm that finds the minimum number of words and the minimum number of letters required to solve an instance of $S$-\textnormal{SIDES} \letterboxed with running time $O(S^2\cdot(n+1)^{S\cdot |\Sigma|}\cdot |D|^2 \cdot L)$.
\end{theorem}
\begin{proof}
    We begin by describing the algorithm that finds the minimum number of words necessary to solve a Letter Boxed puzzle. At the end, we adapt this idea to obtain an algorithm that finds the minimum number of letters. The algorithm, which is based on dynamic programming, runs in $O(S\cdot (n+1)^{S\cdot |\Sigma|}\cdot |D|^2 \cdot L)$ time, where $L$ is the length of the longest word in the dictionary, and uses $O(S \cdot (n+1)^{S\cdot |\Sigma|}\cdot |D|^2 \cdot L)$ memory. 

    Fix an instance of the $S$-SIDES \letterboxed problem. We define a \emph{residual letter state} as an element $r$ of the set $R=(\{0\} \cup [n])^{S\times  \Sigma}$. Each $r\in R$ encodes the number of characters of each kind left on each side of the square, so that $r_{i,\sigma}$ denotes the number of occurrences of the character $\mathtt{\sigma}$ on the $i^{th}$ side. %
    
    We build a weighted directed graph $G=(V,E,c)$ as follows.

    The set of vertices $V$ will encode the states of the game, with two extra states called \texttt{start} and \texttt{end} corresponding to the game being about to start / having ended. In particular we define:
    \[
        V = \bigg(\underset{\substack{\text{the side the}\\\text{ player is}\\\text{currently on}}}{\underbrace{[S]}} \times \underset{\substack{\text{how many occurrences}\\ \text{of each character}\\ \text{are left uncovered}\\ \text{on each side}}}{\underbrace{R}}\times \underset{\substack{\text{the dictionary}\\\text{word that was}\\\text{last started}\\\text{by the player}}}{\underbrace{D}} \times \underset{\substack{\text{how far into}\\\text{the word the}\\\text{player is}}}{\underbrace{[L]}}\bigg) \cup \{\texttt{start}, \texttt{end}\}.
    \]

    One may prefer to think of $V$ as not containing any state $u=(i,r,w,\ell)$ for which $\ell > |w|$, those in which the $i^{th}$ side does not contain an occurrence of the $\ell^{th}$ character of $w$, or those for which $r_{i,\sigma}$ is greater than the original number of occurrences of $\sigma$ on the $i^{th}$ side of the instance. These states will not be reachable from the \texttt{start} state (intuitively because they correspond to unreachable game states) and thus they will not impact the algorithm.
    
    An \emph{initial state} is a state of the form $(i,r,w,1)$ where the $i^{th}$ side of the square contains an occurrence of $w_1$ (the first character of $w$) and $r_{j,\sigma}$ is equal to the number of occurrences of the character $\sigma$ on the $j^{th}$ side of the square at the start of the game, with the exception of $r_{i,w_1}$ which has been diminished by $1$. A \emph{final state} is a state of the form $(i,r,w,\ell)$, where $|w| = \ell$ and the $i^{th}$ side contains an occurrence of the last (the $\ell^{th}$) character of $w$ and $r_{j,\sigma} = 0$ for every $j\in[S]$ and $\sigma\in \Sigma$. We include in $G$ edges of cost zero connecting \texttt{start} to all initial states and edges of cost one connecting all final state to \texttt{end}. In addition, the set of edges $E$ contains two types of edges: \emph{continuation edges} and \emph{new-word edges}.

    Each \textbf{continuation edges} connects two states $u$ and $v$, where in $u$ the player is composing a word $w$ from the dictionary, and in $v$ the player has added another letter to the partial word built so far. In particular, $u = (i,r,w,\ell)$ and $\ell$ is strictly less that the length $|w|$ of the word $w$, and $v = (j,r',w,\ell+1)$, where $j \neq i$, side $j$ contains at least one occurrence of the $(\ell+1)^{th}$ character of $w$, and $r'$ is obtained from $r$ as follows. Let $\sigma = w_{\ell+1}$ be the $(\ell+1)^{th}$ character of $w$ then for all $\tau\in \Sigma$ and $k\in[S]$:
    \[
        r'_{k,\tau} = \begin{cases}
           r_{k,\sigma}-1 &\text{if } k=j\text{, }\tau = \sigma \text{ and }r_{j,\sigma} > 0,\\
            r_{k,\tau} &\text{otherwise.}
        \end{cases}
    \]
    All continuation edges have weight $0$ in $G$.
    
    Each \textbf{new-word edge} connects two states $u$ and $v$ where in $u$ the player has just finished composing a full word $w$ from $D$, and in $v$ the player has played the second letter of another word $w'$ from $D$ that starts with the same letter with which $w$ ends. In particular, $u=(i,r,w,\ell)$ and $v=(j,r',w',2)$, where $w'\in D$ begins with the same character with which $w$ ends, $j \neq i$ and the $j^{th}$ side $\Gamma_j$ contains at least one occurrence of the second character of $w'$, and $r'$ is obtained from $r$ as follows. Let $\sigma$ be the second character of $w'$, then, for all $\tau \in \Sigma$ and $k \in[S]$:
    \[
        r'_{k,\tau} = \begin{cases}
           r_{k,\sigma}-1 &\text{if } k=j\text{, }\tau = \sigma \text{ and }r_{j,\sigma} > 0,\\
            r_{k,\tau} &\text{otherwise.}
        \end{cases}
    \]
    All new-word edges have weight $1$ in $G$.

    Each edge corresponds to a valid transition from one game state to another undertaken by inserting one more character. In particular, the existence of a valid solution corresponds to the existence of a path from \texttt{start} to \texttt{end}. Moreover, the number of words used by a solution corresponds to the weighted length of the corresponding path in $G$, and hence, the original problem is equivalent to finding a shortest path from \texttt{start} to \texttt{end} in $G$.

    This can be done in time linear in the size of the graph by using a version of the BFS algorithm known as 0-1 BFS, which solves the single-source shortest path problem in weighted graphs in which all weights are either zero or one. %

    Note that, by construction, each vertex $u=(i,r,w,\ell)$ with $\ell < |w|$ has out-degree at most $S-1$, and each vertex $u=(i,r,w,\ell)$ with $\ell = |w|$ has out-degree at most $(|S|-1)D$. Hence, the total size of the graph is $O(S^2 \cdot (n+1)^{S\cdot |\Sigma|}\cdot |D|^2 \cdot L)$, and the result follows. 
    
    In order to obtain an algorithm that finds the minimum number of characters needed to solve the Letter Boxed instance, one simply needs to assign to each continuation edge a weight of $1$ instead of $0$. By doing this, the cost of each \texttt{start}-\texttt{end} path in $G$ becomes equal to the number of letters used in the corresponding solution of the puzzle, and the result follows.
\end{proof}

We now use this algorithm to show that \letterboxed is solvable in polynomial time in certain regimes. 

\begin{corollary}\label{cor:letterboxed-poly-cases}
    \letterboxed where either (i) $|\Sigma|=O(1)$, or (ii) $|D|=O(1)$ and $L=O(1)$, or (iii) $L=O(1)$ and $k=O(1)$, is solvable in polynomial time 
\end{corollary}
\begin{proof}
    The puzzle in an instance of \letterboxed has four sides. Thus, when $|\Sigma|=O(1)$, the algorithm of \Cref{thm:letterboxed-main-algorithm} runs in time $O\left((n+1)^{O(1)} \cdot |D|^2 \cdot L\right)$ which is polynomial in $n$, $|D|$, $L$ and thus in the instance size. Recall that the algorithm in \Cref{thm:letterboxed-main-algorithm} returns the minimum number of words required to solve the instance, thus by comparing such value with the input $k$ it is possible to solve the corresponding decision problem.

    Note also that, if $|D|=O(1)$ and $L=O(1)$ one can restrict the alphabet so that $|\Sigma|\leq |D|\cdot L =O(1)$, and this implies that the instance is solvable in polynomial time by the previous case. 

    Finally, if $L=O(1)$ and $k=O(1)$, then it is possible to check whether the puzzle is solvable by brute force; indeed, we can try all $\sum_{i=1}^k |D|^i \leq |D|^{O(1)}$ possible ways of concatenating at most $k$ words from the dictionary, each with an overall length of $O(1)$, and for each of them check in constant time whether it solves the puzzle. Indeed, for each such concatenation we can try all the ways to move around the square (which are at most $4^{O(1)}=O(1)$) and check whether all characters are collected. In particular, note that for $4n> (L-1)k+1=O(1)$, the instance will trivially be unsolvable.
\end{proof}

Sadly, these results do not imply that LETTER BOXED can be solved in polynomial time in general. In fact, we now show that various variants of LETTER BOXED are $\NP$-Complete. We start by proving that the problem is in $\NP$.

\begin{lemma}
    For any $S$, $S$\textnormal{-SIDES} \letterboxed is in $\NP$.
\end{lemma}
\begin{proof}
    We first prove that if a Letter Boxed puzzle is solvable, then it is also solvable with at most $S^2\cdot |\Sigma| \cdot n$ words. We say that a word in a solution is \emph{covering} if it collects at least a new character on one of the sides. Since there are $S \cdot n$ characters to cover, the number of covering words in a solution is at most $S \cdot n$ (if each such word covers only one new character). Moreover, between two covering words, there can be other words that are used to change either the current side or the current character (from which the next word will start from). Since there are at most $S\cdot |\Sigma|$ pairs of current side and last character, any solution can have, without loss of generality, at most $S\cdot |\Sigma|-1$ words between two covering words (if there are more than that, we will necessarily end up in the same state more than once and we can therefore trim the sequence of words). Thus, the total number of words in a solution can be upper bounded by $S\cdot n \cdot S\cdot |\Sigma|$ without loss of generality. 

    Given an instance $(\Sigma, D, \Gamma_1, \dots, \Gamma_S), k$ of $S$-SIDES \letterboxed, a certificate for such an instance is given by a pair $(\texttt{words}, \texttt{sides})$ as defined in \Cref{def:solution-to-letter-boxed}, where $|\texttt{words}| \leq \min\{k, S^2 \cdot |\Sigma|\cdot n\}$ and therefore $|\texttt{sides}|\leq S^2 \cdot |\Sigma|\cdot n \cdot L$. Thus, the certificate has polynomial size and checking the conditions of \Cref{def:solution-to-letter-boxed} can be done in polynomial time. 
\end{proof}

We now prove that LETTER BOXED is $\NP$-Hard even assuming that the dictionary consists of a single word ($|D|=1$) and that the player is required to solve the puzzle with a single word ($k=1$). Specifically, we show a polynomial-time reduction from POSITIVE NOT-ALL-EQUAL 3-SAT, defined as follows.

\begin{definition}[POSITIVE NOT-ALL-EQUAL 3-SAT]
    In POSITIVE NOT-ALL-EQUAL 3-SAT, one is given a set $V$ of boolean variables and a set $C \subseteq \binom{V}{3}$ of clauses each containing exactly three distinct positive literals (non-negated variables). The goal is to decide whether there exists a truth assignment of the variables such that each clause contains both a \emph{true} and a \emph{false} variable.
\end{definition}
POSITIVE NOT-ALL-EQUAL 3-SAT is known to be $\NP$-Complete~\cite{dd20}. We use this fact to show the following.

\begin{theorem}\label{thm:letter-boxed-np-complete-d1-k1}
   \letterboxed is $\NP$-Complete even if $|D|=1$ and $k=1$.
\end{theorem}
\begin{proof}
    We show that the problem is $\NP$-Hard by reducing from POSITIVE NOT-ALL-EQUAL 3-SAT. Let $(V,C)$ be an instance of POSITIVE NOT-ALL-EQUAL 3-SAT. We will first build an instance of LETTER BOXED $(\Sigma, D, \Gamma_1, \Gamma_2, \Gamma_3, \Gamma_4)$ in which the sets $\Gamma_1, \Gamma_2, \Gamma_3, \Gamma_4$ may have different cardinalities, and then show how to pad the instance so that each multiset has the same cardinality. We define the alphabet to be $\Sigma=V \cup C \cup \{\#\} \cup \{*_v \mid v\in V\}$, where $\#$ is a special character, and so is $*_v$ for each $v\in V$. We set $\Gamma_4=\varnothing$ and $\Gamma_1$ contains only the character $\#$ with multiplicity 1 (recall that the same character can be used multiple times). For a variable $v\in V$ let $\eta(v)$ be the number of clauses containing $v$. We set $\Gamma_2=\Gamma_3$ and for each $v\in V$ they contain the characters $v$ with multiplicity $\eta(v)$ and $*_v$ with multiplicity $\eta(v)-1$, while for each $c\in C$ they contain the character $c$ with multiplicity $2$. We are left with defining the only string in the dictionary. We call such string $\sigma$. The string $\sigma$ starts with $\#$, then for each $v\in V$ we append the following string where $v$ is repeated $\eta(v)$ times:
    \begin{equation}\label{eq:letterboxed-1word-p1}
        v\,\, \underbrace{*_v \,\, v \,\, \dots \,\, *_v \,\, v}_{\substack{*_v\, v \text{ repeated} \\ \eta(v)-1 \text{ times}}}\,\, \#.
    \end{equation}
    
    Note that only $\Gamma_2$ and $\Gamma_3$ contain $v$ and $*_v$. Therefore, intuitively, this first string forces the player to choose on which side to collect all the $v$'s and on which side to collect all the $*_v$'s, effectively assigning a truth value to the variable $v$. Then, for each $v\in V$ we append the following string, where $*_v$ is repeated $\eta(v)-1$ times, required for cleanup:
    \begin{equation}\label{eq:letterboxed-1word-p2}
        \underbrace{*_v\,\,\#\,\, *_v\,\,\#\,\, \cdots\,\, *_v\,\,\#}_{*_v\# \text{ repeated $\eta(v)-1$ times}}.
    \end{equation}
    Finally, for each clause $c\in C$ that contains variables $v_1,v_2,v_3$ we append the following string to $\sigma$:
    \begin{equation}\label{eq:letterboxed-1word-p3}
        v_1\,\, c \,\, \#\,\, v_2\,\, c \,\, \#\,\, v_3\,\, c\,\,\# \,\, c \,\,\#.
    \end{equation}
    Intuitively, if all of $v_1, v_2, v_3$ have been first selected from the same side (either $\Gamma_2$ or $\Gamma_3$), it will not be possible to pass through all the four characters $c$ that are present in $\Gamma_2 \cup \Gamma_3$. Otherwise, if at least one of $v_1,v_2,v_3$ was selected from $\Gamma_2$ and at least one from $\Gamma_3$ it will be possible to pass though all the $c$'s. Therefore the clause behaves as a not-all-equal clause.

    We now provide a more detailed argument. We will show that $(V,C)$ is solvable if and only if the instance $(\Sigma ,\{\sigma\}, \Gamma_1, \Gamma_2, \Gamma_3, \Gamma_4)$ is solvable with $k=1$. 

    Suppose $(V,C)$ is solvable and let $\phi:V\rightarrow\{T,F\}$ be an assignment that satisfies all the clauses (in the not-all-equal sense). We show how to follow the string $\sigma$ while collecting all the characters on the square. The player starts on the side $\Gamma_1$ from $\#$. Then, for each $v\in V$, if $\phi(v)=T$ the player moves first to $\Gamma_2$ and then alternates between $\Gamma_3$ and $\Gamma_2$ to process the characters of \Cref{eq:letterboxed-1word-p1} --- if instead $\phi(v)=F$ it moves first to $\Gamma_3$ and then alternates. Similarly, to process the characters of \Cref{eq:letterboxed-1word-p2}, if $\phi(v)=T$ (resp. $\phi(v)=F$) the player alternates between $\Gamma_2$ (resp. $\Gamma_3$) and $\Gamma_1$. At this point the characters we are missing from $\Gamma_2$ (resp. $\Gamma_3$) are those associated with the clauses and those associated with the variables $v$ such that $\phi(v)=F$ (resp. $\phi(v)=T$). Consider the string associated with the clause $c=(v_1, v_2,v_3)$ (\Cref{eq:letterboxed-1word-p3}). We collect $v_1,v_2,v_3$ by going to the side where they still have to be collected (either $\Gamma_2$ or $\Gamma_3$) and collect the $c$ by switching sides. Since the clause is satisfied by $\phi$, at least one of $v_1,v_2,v_3$ will be collected on $\Gamma_2$ and at least one of them will be collected on $\Gamma_3$. Therefore, the first three $c$'s will be collected without repetitions. Thus, thanks to the last portion of the string ($\#\,\, c\,\,\#$) all the four occurrences of $c$ in $\Gamma_2 \cup \Gamma_3$ can be collected. Since this holds for all clauses, at the end of processing $\sigma$, we have collected all characters from $\Gamma_1 \cup \Gamma_2 \cup \Gamma_3$.

    Suppose now that the instance $(\Sigma, \{\sigma\}, \Gamma_1, \Gamma_2, \Gamma_3, \Gamma_4)$ is solvable with a single word $\sigma$. Then, the player is forced to start from $\Gamma_1$ and, for each variable $v$, decide whether to start from $\Gamma_2$ or $\Gamma_3$: if the player starts from $\Gamma_2$ we set $\phi(v)=T$ otherwise we set $\phi(v)=F$. Consider now each portion associated with a clause $c=(v_1, v_2, v_3)$ (\Cref{eq:letterboxed-1word-p3}). Since there are exactly $2\cdot\eta(v)$ characters equal to $v$ in $\sigma$ for each variable $v$, in order for the instance to be solved the player must have taken each of $v_1, v_2, v_3$ on the opposite side that was first used when processing the portion associated with \Cref{eq:letterboxed-1word-p1}. Suppose that all of $v_1, v_2, v_3$ are taken on the same side. Suppose, without loss of generality, that this side is $\Gamma_2$. Then, the first three $c$'s of \Cref{eq:letterboxed-1word-p3} must be taken on $\Gamma_3$. But then at most one of the two $c$ characters can be covered from $\Gamma_2$ and the instance would not be solved. Therefore, under the hypothesis that the instance is solved, for each portion associated with a clause (\Cref{eq:letterboxed-1word-p3}), at least one of $v_1, v_2, v_3$ must be taken on side $\Gamma_2$ and at least one must be taken on side $\Gamma_3$. In other words, for each clause $c=(v_1,v_2,v_3)$, we have $\phi(x)=T$ for at least one $x\in\{v_1,v_2,v_3\}$ and $\phi(y)=F$ for at least one $x\in\{v_1,v_2,v_3\}$. 

    Therefore, we have proved that the problem where $|\Gamma_1|,|\Gamma_2|,|\Gamma_3|,|\Gamma_4|$ can differ is $\NP$-Hard. We now show how to pad the previous instance so that the multisets all have the same cardinality. Consider the previous instance. Add a new character $\tau$. Note that $|\Gamma_2|=|\Gamma_3|>|\Gamma_1|=1$. Add $|\Gamma_2| - 1$ copies of $\tau$ to $\Gamma_1$ and add $|\Gamma_2|$ copies of $\tau$ to $\Gamma_4$. Note that now all multisets have the same cardinality. Finally, append $2|\Gamma_2|-1$ characters $\tau$ to the string $\sigma$. Observe that the previous instance is solvable if and only if this padded instance is solvable. This concludes the proof. 
\end{proof}

We now show that \letterboxed is $\NP$-Complete even when all the words are assumed to have length $5$. This rules out the possibility of designing a polynomial-time algorithm for the case in which $L=O(1)$.

To prove hardness we reduce from the classical 3D MATCHING problem, which has long been known to be $\NP$-Hard~\cite{k09}. We recall the problem definition below.

\begin{definition}[3D MATCHING]\label{def:3d-matching}
    In 3D MATCHING one is given as input three disjoint sets $X\cup Y \cup Z$ with $|X|= |Y| = |Z| = n$ and a collection $T \subseteq X \times Y \times Z$ of triples. The problem asks to decide whether there exists a subset $S \subseteq T$, with $|S|=n$, such that every element in $X\cup Y \cup Z$ appears in exactly one triple in $S$.
\end{definition}

\begin{theorem}\label{thm:letter-boxed-np-complete}
    \letterboxed is $\NP$-Complete even if all words in $D$ have length $5$.
\end{theorem}
\begin{proof}
To show $\NP$-Hardness, we now present a reduction from 3D MATCHING (\Cref{def:3d-matching}). Specifically, consider an instance of 3D MATCHING defined by three disjoint sets $X=\{x_1, \dots, x_n\}$, $Y=\{y_1, \dots, y_n\}$ and $Z=\{z_1, \dots, z_n\}$ and a family $T\subseteq X\times Y\times Z$ of triples. We construct an instance of \letterboxed as follows. We take as alphabet $\Sigma=X\cup Y\cup Z\cup\{\#\}$, where $\#$ is a special character that does not belong to $X\cup Y \cup Z$, and as dictionary
$D\subseteq\Sigma^5\subseteq\Sigma^*$ given by $D=\{\# x_iy_jz_l\# \mid (x_i, y_j, z_l)\in T\}$.
We set $k=n$ and take the four multisets to be:
\[
    \Gamma_1 = \underset{n\text{ times}}{\underbrace{\{\#, \dots , \#\}}}, \hspace{5mm} \Gamma_2 = X, \hspace{5mm} \Gamma_3 = Y, \hspace{5mm} \Gamma_4 = Z.
\]
The instance is illustrated in \Cref{fig:letter-boxed-reduction-1}.
\begin{figure}
    \centering
    \includegraphics[scale=0.238]{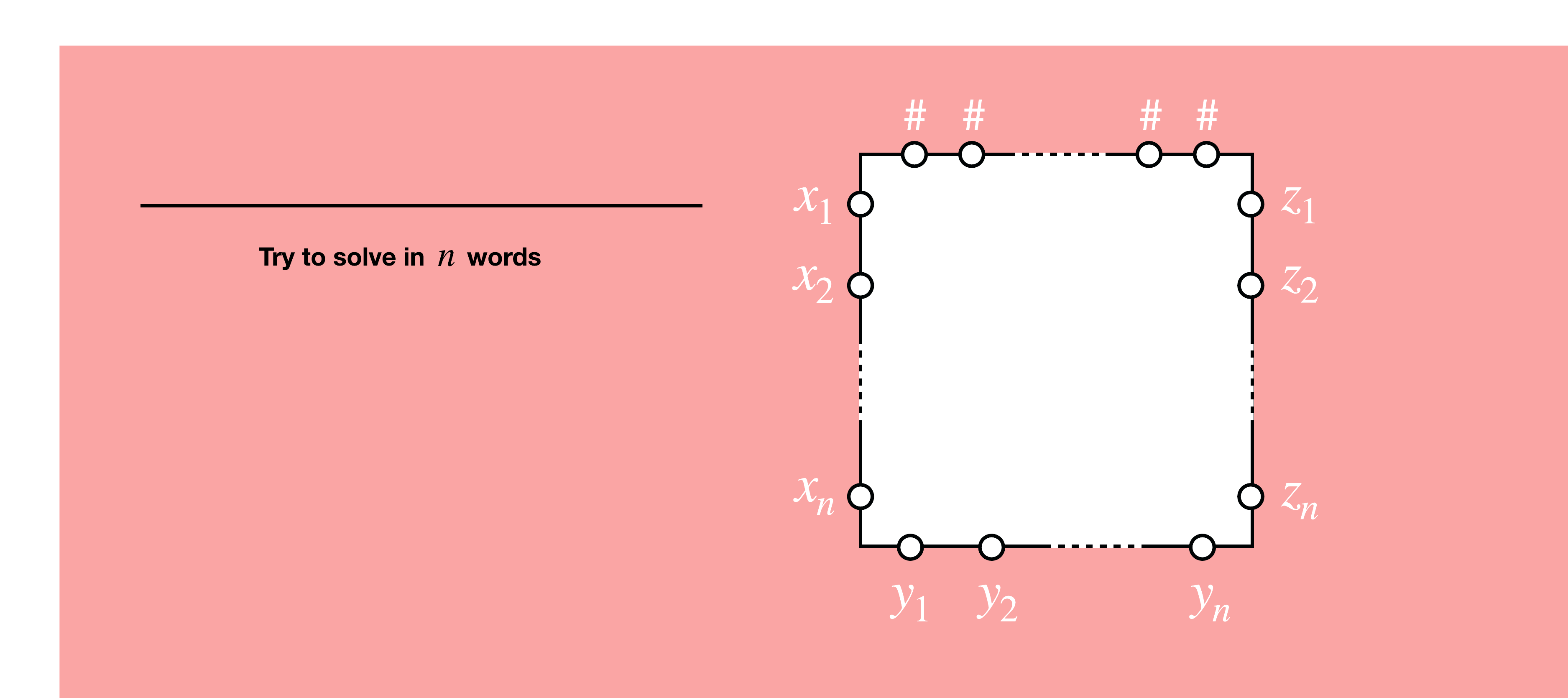}
    \caption{The \letterboxed instance constructed from the 3D MATCHING instance in the proof of \Cref{thm:letter-boxed-np-complete}.}
    \label{fig:letter-boxed-reduction-1}
\end{figure}

We now show that the original 3D MATCHING instance is solvable if and only if the constructed Letter Boxed puzzle is solvable with at most $k$ words. %

Suppose that there exists a subset $S\subseteq T$ such that $|S|=n$ and, for each $v\in X\cup Y\cup Z$, $v$ appears in exactly one triple in $S$. We construct a solution to the Letter Boxed puzzle as follows: order the triples in $S$ arbitrarily, and let $w^{(i
)}$ be the word $\# xyz\#$ where $(x,y,z)$ is the $i^{th}$ triple in $S$ in our arbitrary order. Set $\texttt{words} = (w^{(1)}, \dots , w^{(n)})$ and 
\[
\texttt{sides}= (\underset{1,2,3,4,1 \text{ repeated $n$ times}}{\underbrace{1,2,3,4,1, \dots , 1,2,3,4,1}}),
\]
it is easy to see that $(\texttt{words},\texttt{sides})$ is a valid solution of length $k$ to the Letter Boxed puzzle, since $S$ must have the property that it covers every element in $X \cup Y \cup Z$.

Now, suppose that the Letter Boxed puzzle is solvable with at most $k$ words, as witnessed by a solution $(\texttt{words},\texttt{sides})$, where $\texttt{words}=(w^{(1)}, \dots ,w^{(k)})$. By construction, each $w^{(i)}= \#x^{(i)}y^{(i)}z^{(i)}\#$ for some $(x^{(i)},y^{(i)},z^{(i)}) \in T$. 

Consider the set $S=\{(x^{(i)}, y^{(i)}, z^{(i)}) \mid i=1,\dots, n\}$. Since $\# x^{(i)}y^{(i)}z^{(i)}\#$ is in the dictionary, it follows that $S\subseteq T$. Furthermore, for every $x\in X$ there is one (and only one) $i$ such that $x=x^{(i)}$ for otherwise the Letter Boxed solution would not be able to cover all the letters on every side with only $n$ words. The same applies to $Y$ and $Z$, concluding the proof.
\end{proof}

All our $\NP$-Hardness proofs hold in the standard \letterboxed problem with four sides. We also show that increasing the number of sides can only make the problem harder. 
\begin{theorem}
    For every integer $S \geq 2$, we have:
    \[
        S\textnormal{-SIDES LETTER BOXED} \leq_P (S+1)\textnormal{-SIDES LETTER BOXED},
    \]
    where the starting instance has number of letters on each side $n\geq 2$. Moreover, in the instance with $S+1$ sides, the size of the dictionary $|D|$ increases by at most a constant factor, the maximum length $L$ increases by one, the size of the alphabet $|\Sigma|$ increases by 3, the number of characters on each side increases by one, and the budget of words $k$ increases by $O(S+n)$.  
\end{theorem}
\begin{proof}
    Let $(\Sigma, D, \Gamma_1, \dots, \Gamma_S), k$ be the instance of $S\textnormal{-SIDES LETTER BOXED}$. We build the instance $(\Sigma', D', \Gamma'_1, \dots, \Gamma_{S+1}'), k'$ for $(S+1)\textnormal{-SIDES LETTER BOXED}$ as follows. We introduce three new characters in the alphabet: $\Sigma'=\Sigma \cup \{s, e, \#\}$ such that $s,e,\#\notin \Sigma$. We define the dictionary $D'$ such that: (i) $D\subseteq D'$, (ii) for each $w\in D$, $s \circ w \in D'$ and $w \circ e \in D'$, and (iii) $\#\# \in D'$ and $e\# \in D'$.  

    Let $n=|\Gamma_1|=\dots =|\Gamma_S| \geq 2$. We build $\Gamma'_i=\Gamma_i \cup \{\#\}$ for $i\in [S]$ and $\Gamma'_{S+1} = \{s, e, \#, \dots, \#\}$ where $\#$ has multiplicity $n-1$ in $\Gamma_{S+1}'$. Finally, we set $k'=k+S+1+ 2(n-2)$.

    Suppose that $w^{(1)}, \dots, w^{(\rho)}$ with $\rho\leq k$ is a solution to the instance with $S$ sides (by following the sides in an appropriate way). Then, 
    \[
        s\circ w^{(1)}, w^{(2)}, \dots, w^{(\rho-1)}, w^{(\rho)}\circ e, e\#, \underset{S+2(n-2)\text{ times}}{\underbrace{\#\#, \dots, \#\#}}
    \]
    solves the instance with $S+1$ sides by starting on side $S+1$, following the same solution for the $S$-side case and ending up again on side $S+1$, and finally covering all the $S+n-1$ characters $\#$ with the remaining words. 

    Consider now any solution for the instance with $S+1$ sides. Observe that any solution must start from side $S+1$ with a word of the form $s\circ w$ and must eventually come back to side $S+1$ with a word of the form $w\circ e$. Moreover, any solution for the instance with $S+1$ sides must cover all the characters $\#$ at the end, starting with the word $e\#$ from the side $S+1$, followed by at least $S+2(n-2)$ words $\#\#$ (since the first $S+1$ words can cover one $\#$ character each, and after that two words are required to cover a new $\#$ character). Removing the last words to cover the $\#$ characters, the solution must have the form $s\circ w^{(1)}, \dots, w^{(\rho-1)}, w^{(\rho)}\circ e$ with $\rho \leq k$ and moreover the sequence of words $w^{(1)}, \dots, w^{(\rho)}$ solves the instance with $S$ sides (by following the sides in the same way as in the $S+1$ instance).  
\end{proof}

\section{Pips}\label{sec:pips}
\vspace{-14mm}
\begin{figure}[H]
    \centering    \includegraphics[width=1\linewidth]{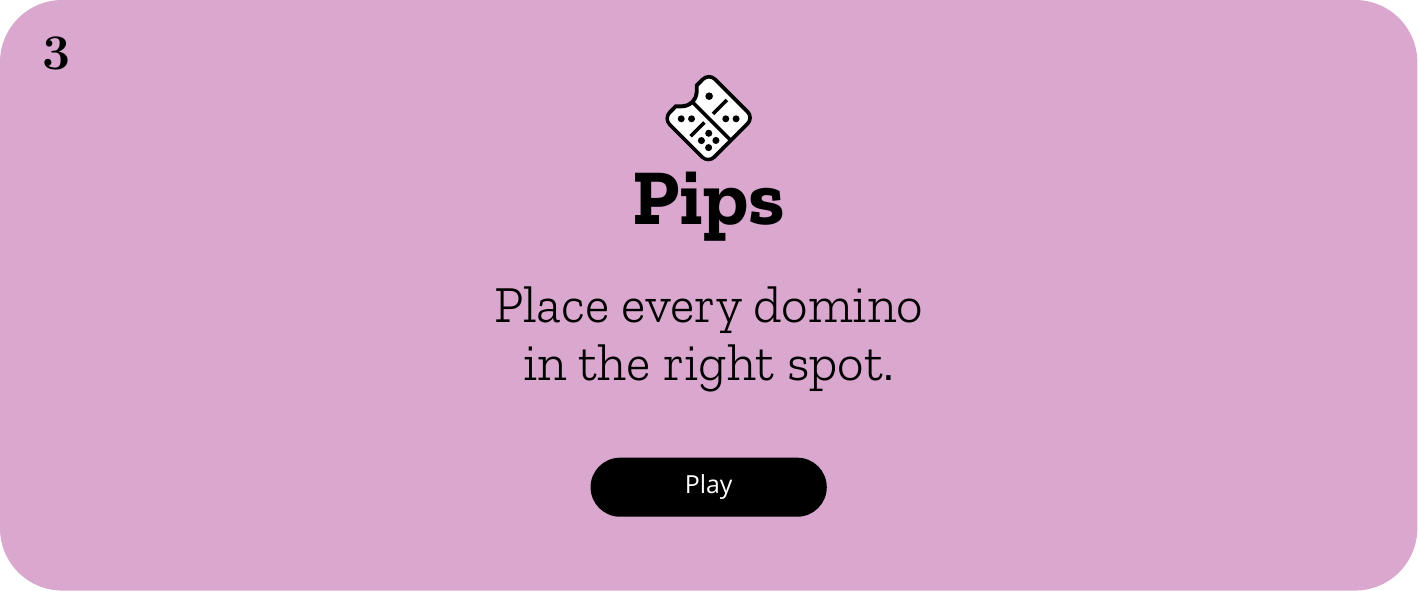}
\end{figure}

In this section, we study the newest of the New York Times games: Pips. This game has a very rich structure, and its expressive set of constraints makes it computationally difficult to find the solution to a Pips puzzle. In particular, in this section, we show that the problem of deciding whether a Pips puzzle has a solution is $\NP$-Complete.

Recall that a Pips puzzle is made up of a board (of arbitrary shape) and some domino tiles. Each tile is divided into two adjacent squares, each of which has a numerical value on it, represented by a certain number of black dots. The board also has constraints in the form of disjoint, connected regions with a constraint indicator signaling that the tile squares placed inside the region have to (i) have the same numerical value as one another (``$=$'' constraints), or have pairwise distinct numerical values (``$\neq$'' constraints), (ii) sum up to a specific numerical value $n$ (``$n$'' constraints), or (iii) sum up to a value that is either greater than or smaller than a specific numerical value $n$ (``$>n$'' and ``$<n$'' constraints). The goal of the game is to place the domino tiles so as to fill the entire board in a way that satisfies all of the constraints specified.%

We consider the following natural problem.

\begin{mdframed}
    \textbf{PIPS}\\
    \textit{Input:} A Pips puzzle, specified by a collection of domino tiles, a board and a set of constraints.\\
    \textit{Question:} Does there exist a valid solution to the instance?
\end{mdframed}

We show that PIPS is $\NP$-Complete and, in particular, we will show a reduction from the POSITIVE PLANAR 1-IN-3-SAT problem, defined below.

\begin{definition}[POSITIVE PLANAR 1-IN-3-SAT]\label{def:pos-planar-1-in-3-sat}
    An instance of POSITIVE PLANAR 1-IN-3-SAT consists of a set of boolean variables and a set of clauses, each containing exactly three distinct positive variables. The objective is to determine a truth assignment to the variables such that each clause contains exactly one true variable. 
    
    The graph associated to the formula is a graph with a vertex for each variable, a vertex for each clause, and an edge between a variable vertex and a clause vertex if the variable appears in the clause. We restrict to instances where such a graph is planar and can be embedded in a \emph{rectilinear} fashion (see \Cref{fig:formulaEmbedding}). That is, with the variables aligned along a straight line and three-legged clauses positioned above and below the variables. The edges between variables and clauses are embedded as straight segments. 
\end{definition}

POSITIVE PLANAR 1-IN-3-SAT is known to be $\NP$-Hard even when restricted to rectilinear instances \cite{mr08}.

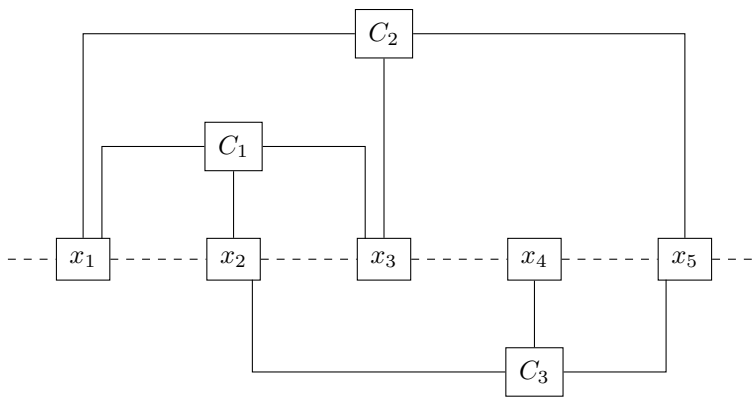
\begin{figure}
\begin{center}
\begin{tikzpicture}[ scale=1,
  node/.style={rectangle, draw, inner sep=5pt, font=\sffamily\bfseries},
  solid edge/.style={-},               
  dashed edge/.style={-, dashed}
]

\node[node] (1) at (0,0) {$x_1$};
\node[node] (2) at (2,0) {$x_2$};
\node[node] (3) at (4,0) {$x_3$};
\node[node] (4) at (6,0) {$x_4$};
\node[node] (5) at (8,0) {$x_5$};

\node[node] (c1) at (2,1.5) {$C_1$};
\draw[solid edge] (0.25,0.272) -- (0.25,1.5) -- (c1);
\draw[solid edge] (2) -- (c1);
\draw[solid edge] (3.75, 0.272) -- (3.75,1.5) -- (c1);

\node[node] (c2) at (4,3) {$C_2$};
\draw[solid edge] (1) -- (0,3) -- (c2);
\draw[solid edge] (3) -- (c2);
\draw[solid edge] (5) -- (8,3) -- (c2);

\node[node] (c3) at (6,-1.5) {$C_3$};
\draw[solid edge] (2.25,-0.272) -- (2.25,-1.5) -- (c3);
\draw[solid edge] (4) -- (c3);
\draw[solid edge] (7.75, -0.272) -- (7.75,-1.5) -- (c3);

\draw[dashed edge] (-1,0) -- (1);
\draw[dashed edge] (1) -- (2);
\draw[dashed edge] (2) -- (3);
\draw[dashed edge] (3) -- (4);
\draw[dashed edge] (4) -- (5);
\draw[dashed edge] (5) -- (9,0);
\end{tikzpicture}
\caption{Rectilinear embedding of the graph associated with the boolean 1-in-3-SAT formula $(x_1 , x_2, x_3) \land (x_1, x_3, x_5) \land ( x_2,  x_4, x_5)$.}
\label{fig:formulaEmbedding}
\end{center}
\end{figure}

\begin{theorem}[\Cref{thm:pips-overview} paraphrased]
    \textnormal{PIPS} is $\NP$-Complete, even if there are only two types of domino tiles (one containing two zeroes, and one containing two ones) and only constraints of type ``='' and ``n''.
\end{theorem}
\begin{proof}
    Given an assignment of the domino tiles to the positions on the board, one can first ensure this assignment tiles the board exactly, then simply iterate over all the constraints and check they are satisfied. This can be done in time polynomial in the size of the instance. Hence PIPS $\in \NP$.

    We now show that PIPS is $\NP$-Hard by reducing from POSITIVE PLANAR 1-IN-3-SAT (\Cref{def:pos-planar-1-in-3-sat}). We give a reduction where all the domino tiles are of two types: either they have a value of $1$ in both squares (one-tiles) or they have a value of $0$ in both squares (zero-tiles), and the instance we construct will only make use of two types of constraints: ``='' constraints and ``$n$'' constraints, implying that even this restricted version of the game is $\NP$-Hard.
    
    The construction is based on gadgets. We introduce three types of gadgets: variable gadgets (one for each variable in the POSITIVE PLANAR 1-IN-3-SAT instance), clause gadgets (one for each clause in the POSITIVE PLANAR 1-IN-3-SAT instance) and cleanup gadgets. 
    
    In the Pips puzzle we construct, there are exactly enough zero-tiles to cover all the clause and variable gadgets, and enough one-tiles to cover all the variable gadgets. The size of the cleanup gadgets is chosen so that the area of the board, in squares, equals twice the total number of tiles. We now describe the three different types of gadgets in detail.

    \paragraph{Variable gadgets.} Variable gadgets (\Cref{fig:pips-variable-gadget}) are regions of the board made up of a straight horizontal line (the base) and vertical lines (branches) that leave the base and connect to clause gadgets. Each variable gadget is connected to all the clause gadgets corresponding to clauses that the variable belongs to. The length of the base is chosen to allow the number of branches to equal the number of clauses that the variable appears in. The length of the branches is constructed to be always even, so that no domino piece can overlap with the gadget without being fully contained in it. This is done simply by placing the clause gadgets at an even vertical distance from the variable gadgets.

    \begin{figure}[H]
        \centering
        \includegraphics[width=1\linewidth]{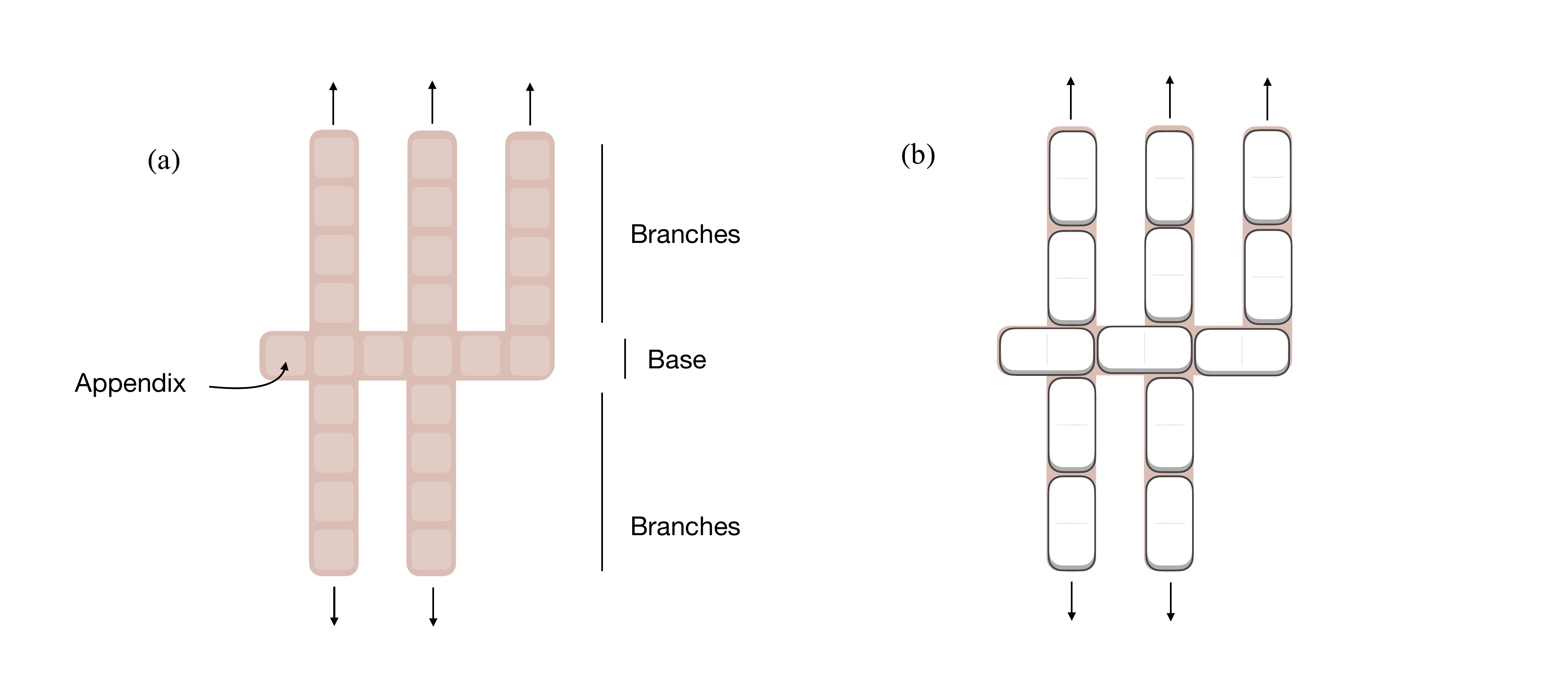}
        \caption{A variable gadget connected to three clause gadgets above and two below. The branches continue towards the corresponding clause gadgets. On the left, the structure of the gadget. On the right, we illustrate the only possible way to tessellate the variable gadget with $2\times 1$ domino tiles.}
        \label{fig:pips-variable-gadget}
    \end{figure}

    Each variable gadget also has a constraint (\Cref{fig:pips-variable-gadget-constraint}) on it, imposing that all the squares in the gadget, with the exception of the last square on every branch, must be occupied by domino pieces of the same value. Since there is a unique way to tessellate the variable gadgets, and this involves having a piece which overlaps with the constraint in only one square, all the squares in the gadget (including the tips of the branches) are forced to be covered by the same number. Note that we exclude the last square on every branch from the constraint to allow for it to be part of the constraint in the clause gadget (see below). 

    \begin{figure}[H]
        \centering
        \includegraphics[scale=0.4]{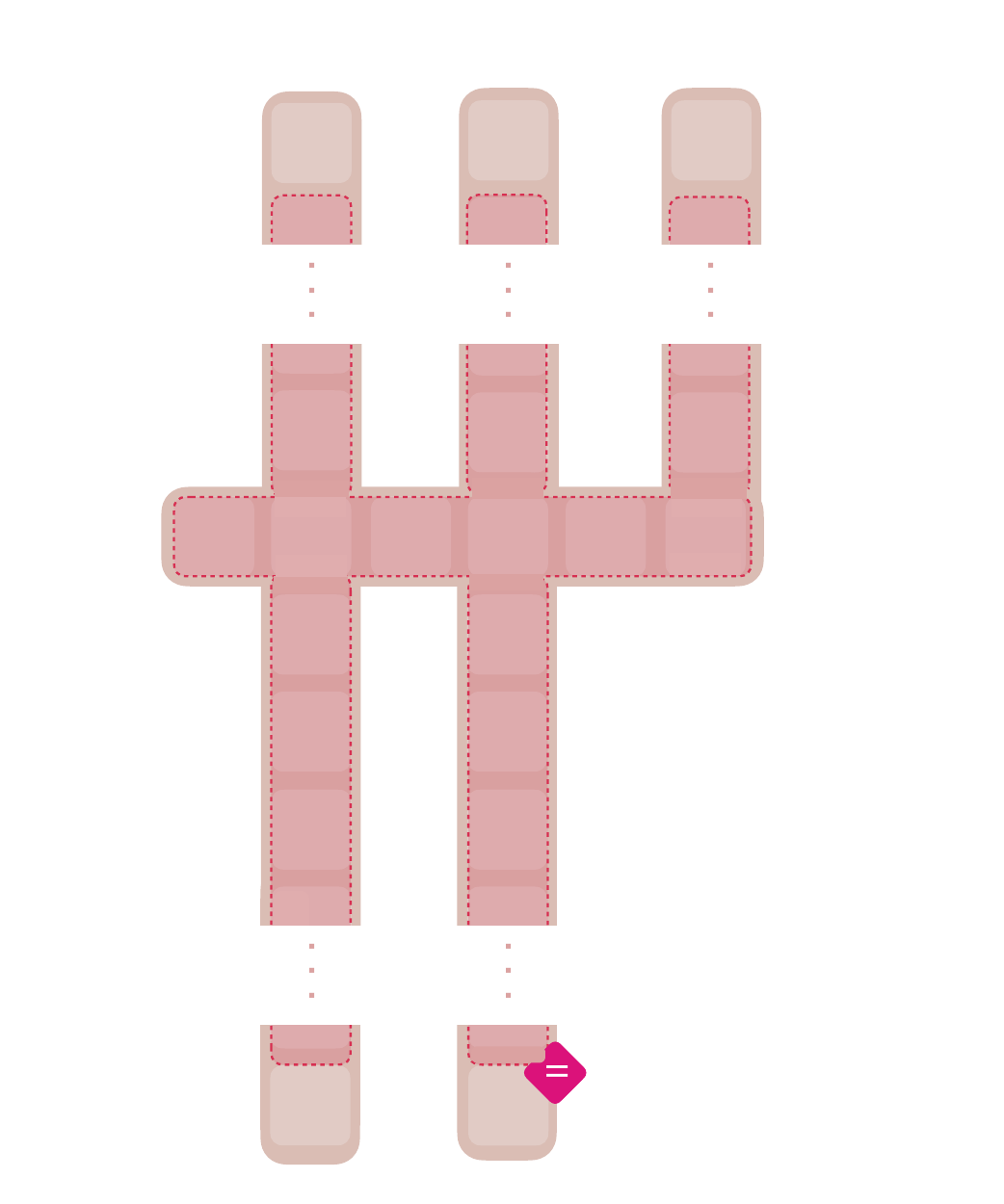}
        \caption{The constraint on a variable gadget.}
        \label{fig:pips-variable-gadget-constraint}
    \end{figure}

    \paragraph{Clause Gadgets.} Each clause gadget (\Cref{fig:clause gadget}) is a horizontal region of the board that contains the ends of three branches belonging to three different variable gadgets: the ones corresponding to the variables taking part in the clause. These branches all end at the same \emph{latitude} (the same y coordinate) and are connected by a horizontal line: the body of the clause. In order to ensure that the body of the clause can be tiled by $2\times 1$ domino tiles, we make the following small adjustment: if two consecutive branches entering a clause gadget are at odd horizontal distance from each other, the horizontal line connecting them contains a vertical domino piece.

    Clause gadgets also each have a constraint on them. This requires that the values covering the entire gadget, including the ends of the variable gadget branches, have a total sum of 1. Note that by construction there is only one way to tile the squares in the clause gadgets that do not correspond to ends of variable gadgets.

    \begin{figure}[H]
        \centering
        \includegraphics[width=\linewidth]{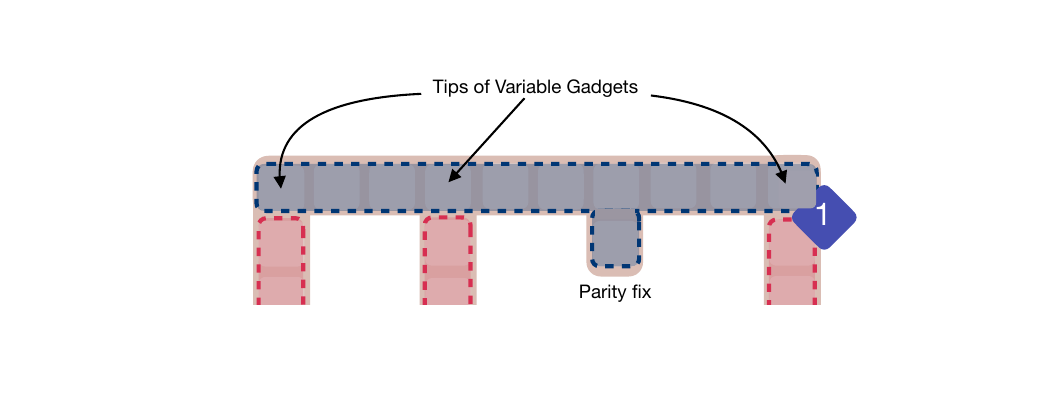}
        \caption{A clause gadget for a clause above the variables. Clause gadgets for clauses below the variables are mirrored vertically.}
        \label{fig:clause gadget}
    \end{figure}
    
    \paragraph{Cleanup Gadgets.} The board then contains cleanup gadgets: regions of the board made up of straight lines with no constraints that serve the purpose of collecting the tiles that were not utilized to cover the other gadgets.

    We remark the following: by construction all the possible valid solutions to the constructed instance are the same up to relabeling the numbers on the domino tiles. In particular, there is only one way to place $1\times 2$ dominoes to tile all the elements of a variable gadget, and in turn this implies that the ends of the variable gadgets inside the clause gadgets cannot be covered by a tile that overlaps with the rest of the clause gadget. 

    This immediately implies that the body of each clause gadget must be covered entirely by tiles with value zero (by the constraint on the clause gadget), and each variable gadget is covered either entirely by tiles with value $1$ or entirely by tiles of value $0$ (by the constraint on the variable gadgets). These correspond to choosing a truth assignment of the variables in the original POSITIVE PLANAR 1-IN-3-SAT instance.

    Moreover, the constraints in the clause gadget impose that the puzzle is solved if and only if, for each clause, exactly one of the variable gadgets is covered by tiles of value $1$. 

    Hence the puzzle is solvable if and only if the original POSITIVE PLANAR 1-IN-3-SAT instance was solvable.

    Since the construction is polynomial in size and can be carried out in polynomial time this yields the wanted result.
\end{proof}

In the construction above, the board might consist of disconnected regions. We note that, while this is technically allowed by the game, it is not necessary for our reduction. One can easily have the reduction produce an instance consisting of a single connected region by adding extra horizontal tiles (``connection gadgets'') connecting the variable gadgets and the cleanup gadget, and put constraints on these regions so that each $2\times 1$ section must sum up to 0 and add just enough extra zero-zero tiles to the instance to cover the connection gadgets. It is easy to see that the previous instance is solvable if and only if this new connected instance is solvable.

In the New York Times version, the only numbers that can appear on Pips tiles are numbers from $0$ to $6$. We note that if one can assign arbitrary non-negative integer values to the domino tiles, then Pips becomes (weakly) $\NP$-Hard even restricted to instances with a single ``$n$'' constraint.

\begin{figure}
    \centering
    \includegraphics[width=\linewidth]{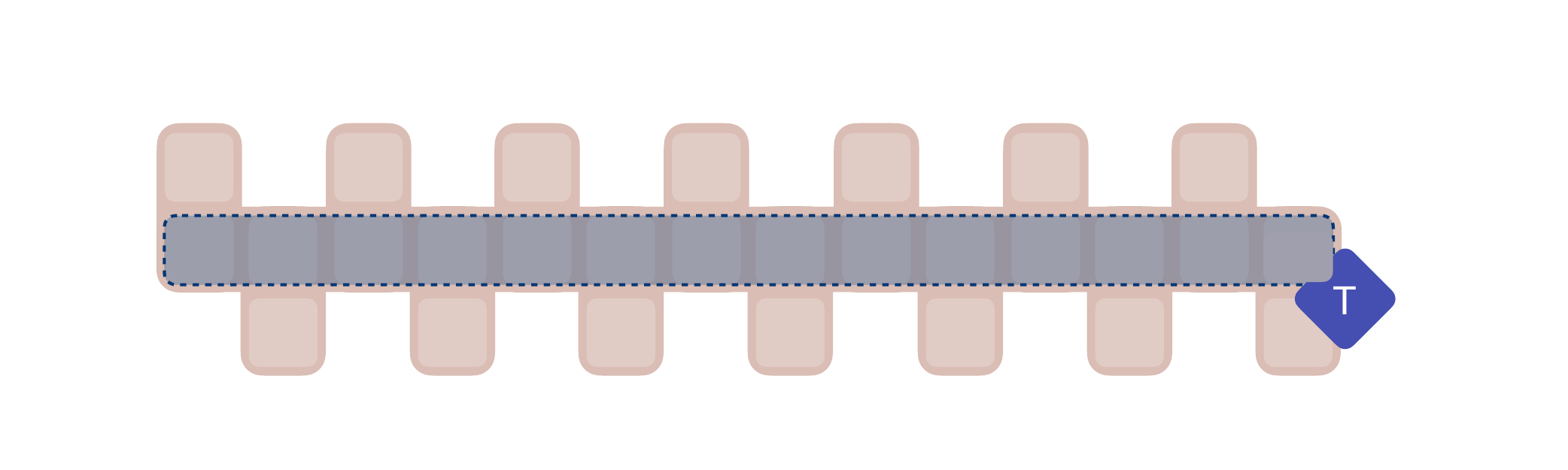}
    \includegraphics[clip, trim=0cm 1cm 0cm 1cm, width=0.5\linewidth]{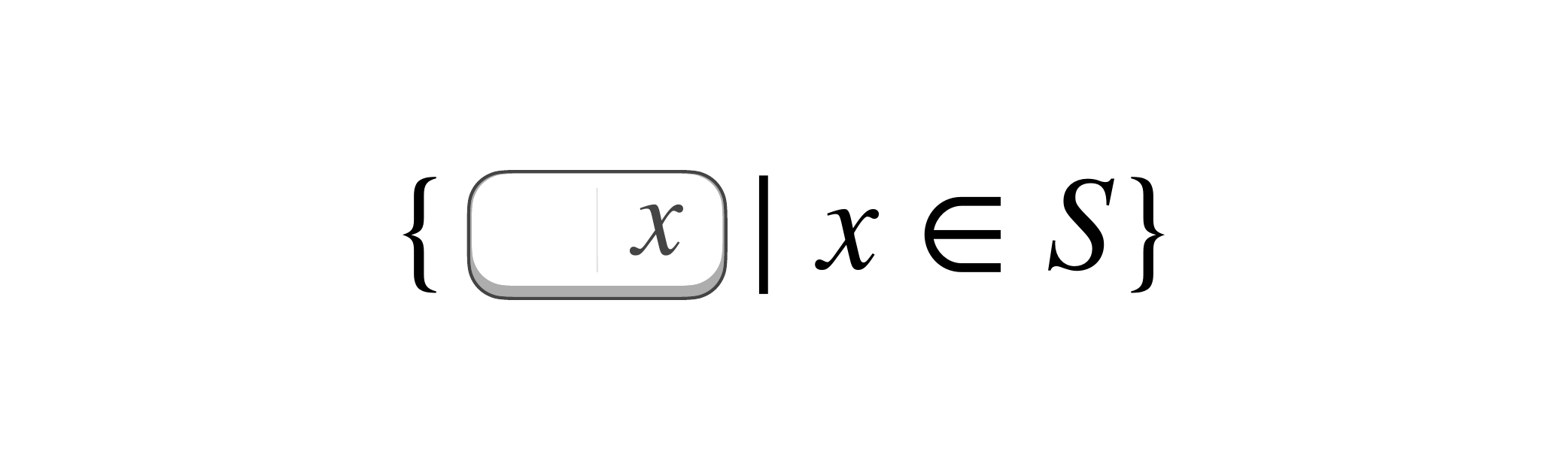}
    \caption{The instance of PIPS constructed in the proof of \Cref{thm:pips-subset-sum}. Here, $T$ is the target sum of the SUBSET SUM instance we are reducing from, and $S$ is the (multi)set of elements of the SUBSET SUM instance.}
    \label{fig:pips-weakly-np-hard}
\end{figure}

\begin{theorem}\label{thm:pips-subset-sum}
    PIPS, with arbitrary non-negative integer numbers on tiles, and a single ``$n$'' constraint is weakly $\NP$-Hard.
\end{theorem}
\begin{proof}
    We show this by giving a reduction from SUBSET SUM. Recall the SUBSET SUM problem:
    \begin{mdframed}
    {SUBSET SUM (with positive integers)}\\
    \textit{Input:} A multiset $S$ of positive integers and a target sum $T \in  \mathbb{N}$. \\
    \textit{Question:} Is there a sub(multi) set $S'\subseteq S$ such that:
    \[
        \sum_{x\in S'}x = T?
    \]
\end{mdframed}
    This problem is well known to be weakly $\NP$-Hard~\cite{k09}.

    Given an instance of SUBSET SUM defined by a multiset $S$ and a target sum $T$, we construct an instance of PIPS which includes, for each $x \in S$ a tile that has value $x$ on one square, and $0$ on the other. The board is then constructed as in \Cref{fig:pips-weakly-np-hard}.

    The PIPS instance contains a single ``$T$'' constraint that spans the horizontal block of length $n$, where $n$ is the cardinality of $S$ (with multiplicity), and that therefore enforces the sum of the dominoes to be $T$.

    Every tessellation of the board must then choose, for each $x\in S$ whether to include $x$ in the horizontal block, or whether to include $0$ instead. In the end, the solution is valid if and only if the sum of all the values included in the central horizontal line equals $T$. Hence, the PIPS instance constructed has a solution if and only if the SUBSET SUM instance is solvable. Since the reduction is clearly polynomial in size and can be carried out in polynomial time, the result follows.
\end{proof}

\section{Strands}\label{sec:strands}
\vspace{-14mm}
\begin{figure}[H]
    \centering
    \includegraphics[width=1\linewidth]{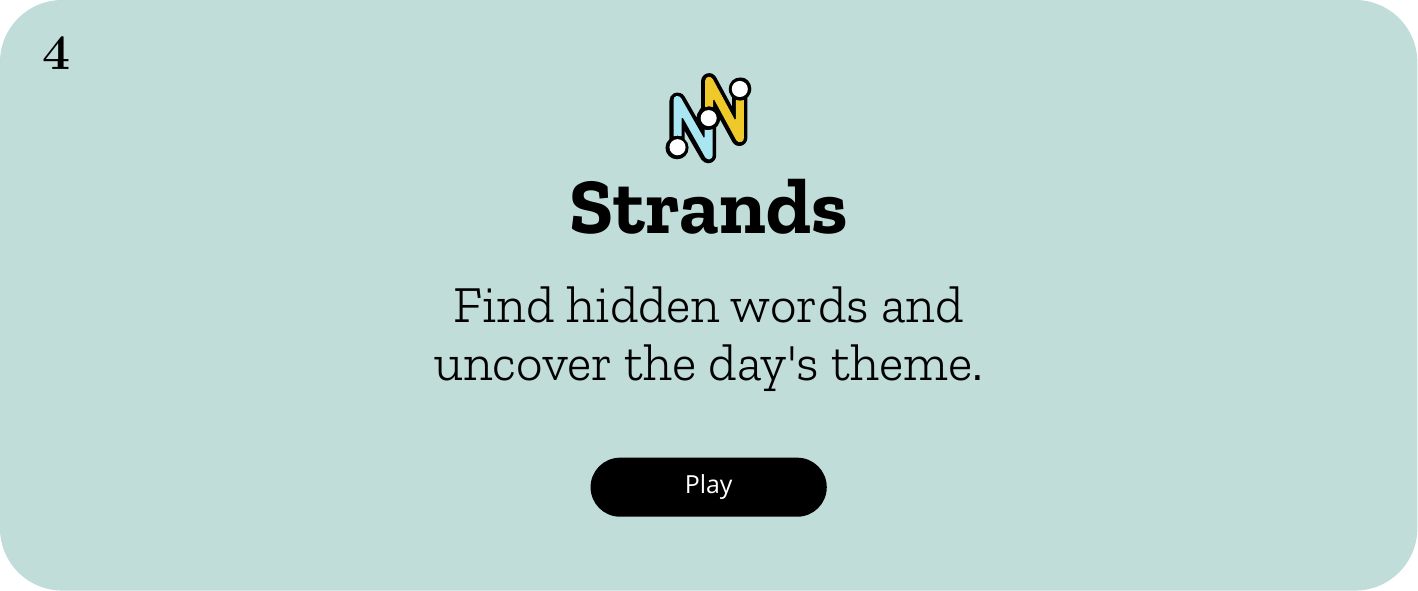}
\end{figure}

In this section we present our results on Strands. Recall that, in this game, one is given a grid of characters and the goal is to cover the grid with disjoint valid words (in the New York Times version, the player receives hints from the game in an interactive way). We will focus on the problem of finding a valid partition of the grid into words from the dictionary (according to the adjacency rules of the game) and hence deciding whether a puzzle is solvable. We consider the simplest instantiation of this problem, in which there is no interactivity between the player and the game, and the player simply needs to decide, given the dictionary and the grid, whether at least one possible solution exists.

We define an instance of Strands as follows.

\begin{definition}[Strands Instance]
    A \emph{Strands instance} is a tuple $(\Sigma,D,M)$, where $\Sigma$ denotes a finite alphabet, $D \subseteq \Sigma^*$ is a finite set of strings over $\Sigma$, referred to as the \emph{dictionary}, and $M \in \Sigma^{n\times m}$ is a grid of characters from $\Sigma$ with $n$ rows and $m$ columns.
\end{definition}

We consider the following decision problem. Given an instance $(\Sigma, D, M)$, decide whether $M$ can be partitioned into a sequence of words $w^{(1)},\dots, w^{(k)}$ such that $w^{(i)}\in D$ for each $i$. Here, each word is required to appear as a sequence of adjacent characters according to the Strands adjacency rule: when reading off a word, the next character can appear in any of the eight cells neighboring the current character (in particular, diagonal movements are allowed). Note that a contiguous sequence of characters on the grid $M$ can be read in any direction, and that each word can only use each entry of the grid $M$ at most once.

Our main result in this section is that this problem is $\NP$-Complete even with a constant-size dictionary of constant-length words. 
\paragraph{On Spangrams.} Recall that, in the New York Times' version of the game, every puzzle contains a \emph{spangram}: a word that connects two opposite sides of the grid. In general, requiring the game to have a spangram, rules out the possibility of having the maximum word length $L$ bounded by a constant, since as the grid grows, the minimum length of a spangram would grow with it. So in order to prove our lower bound in the stronger setting of $L=O(1)$, we shall not impose that instances contain a spangram. We note, however, that adding a spangram, e.g.\ at the bottom of the instance, does not make the problem easier, and hence our hardness results will transfer to the set of instances that do contain a spangram via a straightforward reduction.
\paragraph{A simple reduction from Flow Free.} The Strands problem, as defined above, shares some similarities with the \emph{Flow Free} (or \emph{Zig-zag Numberlink}) game, which was shown to be $\NP$-Hard \cite{addorvs15}. A Flow Free instance consists of $k$ terminal pairs---each pair represented by the same color---placed in an $n \times m$ grid. A solution to the instance is defined as a set of $k$ disjoint non-diagonal paths, each connecting the two terminals of a pair and collectively covering all vertices of the grid. (We refer the reader to \cite{addorvs15} for more details on the game)
One can prove NP-Hardness of Strands by giving a polynomial-time reduction from Flow Free: given a Flow Free instance with colors $C_1,\dots,C_k$ and grid size $n \times m$, consider a $n \times m$ Strands grid with $\Sigma = \{C_1,\dots , C_k, B, W\}$, where each cell containing a terminal $C_i$ is filled with the letter $C_i$, while the remaining cells are filled with black $B$ and white $W$, arranged in a checkerboard pattern. Define the dictionary $D$ as $D = \bigcup_{i=1}^k D_{C_i}$ where each $D_{C_i}$ corresponds to the possible paths connecting pair of terminals of color $C_i$ and is given by
\[
D_{C_i} =
\begin{cases}
\{ C_i \,(B\,W)^l\, B\, C_i \mid 0 \leq l \leq m \cdot n \} & \text{if both terminals lie on white cells};\\[6pt]
\{ C_i\, (W\,B)^l\, W\, C_i \mid 0 \leq l \leq m \cdot n \} & \text{if both terminals lie on black cells};\\[6pt]
\{ C_i\, (B\,W)^l\, C_i \mid 0 \leq l \leq m \cdot n \} & \text{if the terminals lie on cells of different colors}.
\end{cases}
\]
It is then not hard to check (but also not all that important, given the stronger results below) that this reduction is correct.

However, observe this reduction from Flow Free requires a polynomially large alphabet, dictionary, and word lengths, and hence one needs a different reduction to prove the stronger statement we aim for. 

\smallskip

We now prove the main theorem of this section.

\StrandsNPComplete*
\begin{proof}
    We start by showing that the problem is in $\NP$. A certificate consists of two matrices. The first is $V_1\in\{S,E,C\}^{n\times m}$ and the second is $V_2\in \{u,r,l,d,ur,ul,dr,dl\}^{n\times m}$. Intuitively, $V_1$ indicates for each cell whether it is the start, end, or continuation of a word. Instead, $V_2$ indicates the direction to move from the current cell to the next one to build a word. Given these two matrices, one can verify in polynomial time whether they encode a partition of $M$, and for each word in that partition, ordered according to the directions given by $V_2$, one can check whether the word is in $D$. Thus, the problem belongs to $\NP$.

    To establish that \emph{Strands} is $\NP$-hard, we provide a reduction from POSITIVE PLANAR 1-IN-3-SAT with the restrictions described above (see \Cref{def:pos-planar-1-in-3-sat}). At a high level, given a boolean formula, our reduction substitutes each variable vertex in the corresponding rectilinear embedding of the graph with a vertex gadget, each clause vertex with a clause gadget, and each edge with an edge gadget (see \Cref{fig:strands-example-reduction} for an example). %

    More precisely, the Strands instance of our reduction has alphabet and dictionary equal to:
    \[
        \Sigma = \{A,B,C,\,*,\,\#,E,F,\sqcup\},
    \]
    \[
        D = \{\sqcup,\,A,\,A*,\,*\,\#BE,\,\#B*,\,EF,\,, FCC\}.
    \]
    We now specify how to construct the grid $M$ of the Strands instance. It will be useful to separate the various gadgets with ``blank'' cells. Each such cell is represented by the character $\sqcup \in \Sigma$, which is also the only valid word in $D$ containing this character. Consequently, the ``blank'' cells can always be collected in a unique way, and the character $\sqcup$ will not be displayed for simplicity.

\begin{figure}
    \begin{center}
        \includegraphics[width=\textwidth]{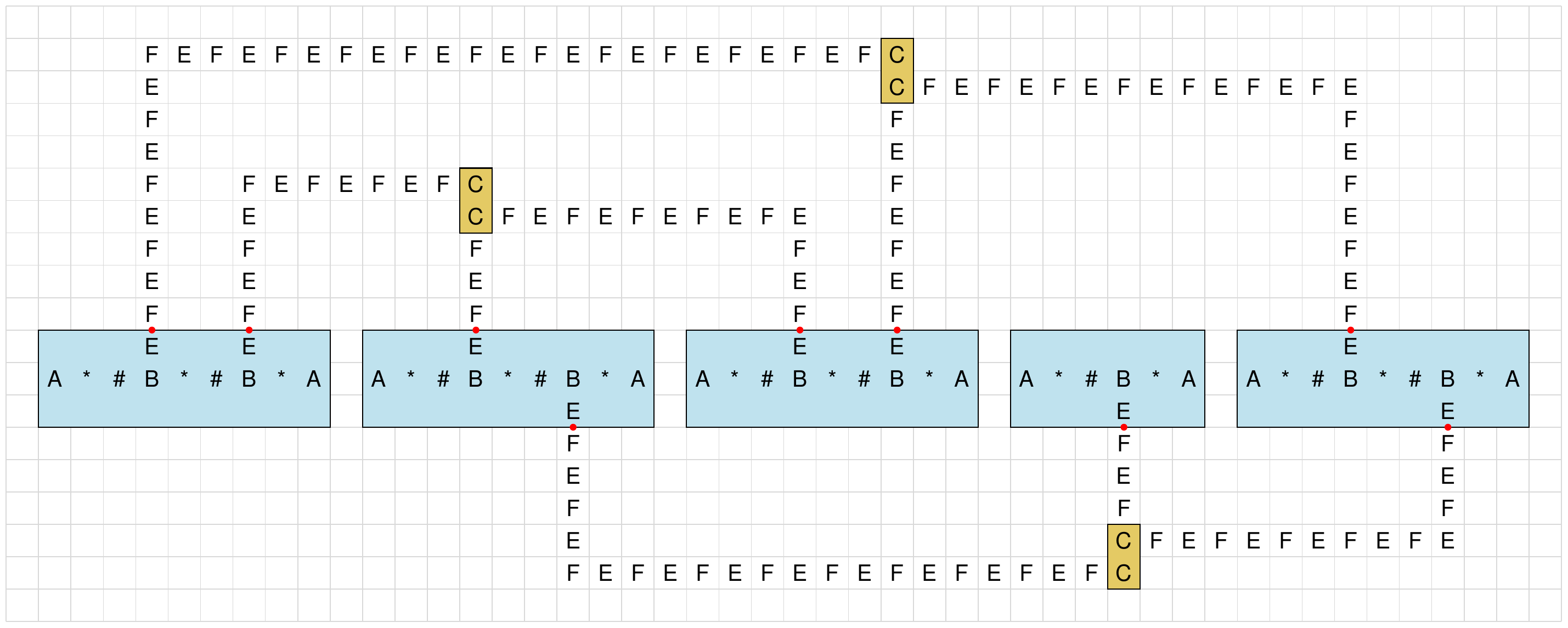}
        \caption{Strands instance corresponding to the POSITIVE PLANAR 1-IN-3-SAT instance given by the formula
        $(x_1, x_2, x_3) \land (x_1, x_3, x_5) \land ( x_2, x_4 , x_5)$. Variable gadgets are highlighted in blue, while clause gadgets are highlighted in yellow.}
        \label{fig:strands-example-reduction}
    \end{center}
\end{figure}

    \paragraph{Vertex gadget.} %
    Each variable $x_i$ is depicted as a $3\times (3k_i+3)$ rectangle, where $k_i$ denotes the number of clauses in which $x_i$ appears, as displayed in \Cref{fig:strands-variable-gadget}. 
    
    \begin{figure}[h!]
    \centering
    \includegraphics[width=0.6\textwidth]{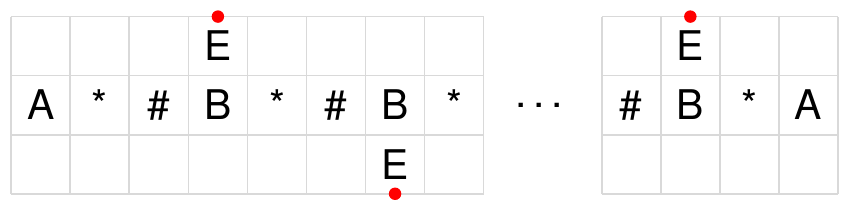}    
    \caption{Variable gadget associated to a variable vertex. Each $E$ character is the starting point of an edge (we highlight this with a red dot).}
    \label{fig:strands-variable-gadget}
    \end{figure}

    Specifically, the first and second columns of the gadget are identical for every variable, followed by a central $3\times 3$ module, which is repeated as many times as the number of clauses in which the corresponding variable appears. In the $k^{th}$ module, the character $E$ represents the starting point of the $k^{th}$ edge connecting the variable to a clause (the order is determined by scanning the edges of the rectangular embedding along the line of variables, from left to right). This character may appear either below or above the corresponding $B$, depending on whether the $k^{th}$ edge leaving the variable node starts from the bottom or the top of the square. Finally, the last column of the gadget is equal to the first one.

    Note that the character $A \in \Sigma$ appears only in words $A \in D$ and in $A\,* \in D$, whereas the character $*$ appears only in  $A\,* \in D$, $*\,\#\,BE \in D$, and $\#\,B* \in D$. Therefore, if the first column of the gadget associated with variable $x_i$ is covered using word $A$, then the first occurrence of $*$ must be covered with $*\,\#\,BE$. Consequently, the remainder of the gadget must be covered exclusively using $*\,\#\,BE$ and $A\,*$.    Conversely, if the first column is covered using $A\,*$, then the rest of the gadget must be covered only with $\#\,B*$ and word $A$. The first method covers all cells of the gadget and will be interpreted as the true assignment for the variable $x_i$, whereas the second method leaves the $E$'s uncovered and corresponds to a false assignment.

    \paragraph{Clause gadget.} Each clause $(x_i,x_j,x_k)$ is associated to a clause gadget of size $2\times 1$ consisting of exactly two characters $C\in \Sigma$. Assuming that variables $x_i$, $x_j$ and $x_k$ are sorted according to their order in the embedding of the graph, the clause gadget will be placed on the same column of one of the $E$ characters exiting from variable $x_j$, leaving an odd number of rows between the $E$ of the variable gadget and the first $C$ of this gadget. Thanks to the rectilinear embedding, the edge exiting from $x_i$ will arrive from the left of the clause, the one exiting from $x_k$ will arrive from the right, and that exiting from $x_j$ will arrive from below or above depending whether the clause is above or below the line of variables. %

    \paragraph{Edge gadget.} Each edge connecting a variable to a clause is represented by an alternating sequence of $E$'s and $F$'s of even length, starting from an appropriate $E$ in the variable gadget. Thanks to the rectilinear embedding of the graph, each path is either vertical, or it consists first of a vertical segment and then of an horizontal segment (going left or right).  

    Depending on the position of the first $E$, the sequence will connect to one of the two $C$'s of the clause gadget. Specifically, if the Manhattan distance between the $E$ in the variable gadget and the closest $C$ in the clause gadget is even, the even-length sequence will connect to the closest $C$, otherwise it will connect to the other $C$ in the clause gadget.

    Thus, every edge starts with the character $E$ inside the variable gadget and ends with the character $F$ immediately adjacent to the clause gadget. Since the graph is planar and every gadget is well-separated from the others, if the false assignment is chosen, then every edge gadget leaving $x_i$ must be covered exclusively using $EF$, until reaching and covering the last occurrence of $F$ along that edge path. Conversely, if the true assignment is chosen for the variable $x_i$, then all the $E$'s of the gadget associated with $x_i$ have already been covered, and the first $F$ of an outgoing edge must be covered using the word $EF$ together with the second occurrence of $E$. Therefore, the entire path must be covered using $EF$ until reaching the last $F$, which remains uncovered this time. This last $F$ must then be covered using the word $FCC \in D$, thereby covering the whole clause gadget.
    
    Note that it could be possible to move diagonally from the last $F$ of an edge gadget to the last $F$ of another one; however, there is no word in $D$ containing two $F$'s. Therefore, the edge gadgets cannot be exploited to cheat.\\

    Note that the width of the grid is linear in the sum of the degrees of the variable vertices, and it is thus linear in the number of clauses. The height is also linear in the number of clauses since each clause gadget has constant height. Thus, the grid has polynomial size and the reduction can be carried out in polynomial time.
    
    We now aim to show that the original formula is satisfiable (in 1-in-3-SAT sense) if and only if the corresponding Strands instance is solvable.

    If the formula is satisfiable, consider a satisfying truth assignment for the variables and cover the corresponding gadget associated with each variable according to the truth assignment. Then, cover every edge gadget according to the unique covering determined by the truth assignment of the variable, as described previously. Let $C_k$ be a clause of the formula. Since the truth assignment satisfies the formula, there exists exactly one true variable $x_i$ appearing in $C_k$. Consequently, the corresponding clause gadget is fully covered by the covering of the edge connecting $x_i$ and $C_k$, since it covers the last occurrence of $F$ using $FCC$. Thus, every gadget associated with a clause is completely covered. Since no other true variable appears in $C_k$, all other edges entering $C_k$ connect a false variable to $C_k$. Hence, they must have been fully covered using exclusively $EF$, without using the $C$’s of $C_k$. Therefore, each edge gadget has been fully covered. In conclusion, every gadget is completely covered; hence, the Strands instance is solved.   
    \smallskip
    If the corresponding Strands instance is solvable, consider a valid partition of the grid $M$. Since, as previously described, there are only two ways to cover the variable gadgets, this solution encodes a valid truth assignment to the variables. Let $C_k$ be a clause. Since each gadget associated with a clause is fully covered, there exists an edge entering $C_k$ which is fully covered using $FCC$; in other words, the edge connects $C_k$ to a true variable. Therefore, there is at least one true variable appearing in $C_k$. By contrast, suppose that at least two true variables appear in $C_k$. Consider their respective edges connecting to $C_k$: both connect a true variable to $C_k$ and both must be fully covered, therefore both would use the two $C$’s of $C_k$, which is impossible. Therefore, there is exactly one true variable for each clause and the truth assignment satisfies the formula. \qedhere
\end{proof}
While diagonal movements are allowed in Strands, this does not make any difference in our reduction, in fact:
\begin{corollary}
Given a Strands instance $(\Sigma, D, M)$ with $|\Sigma|, |D|, L = O(1)$ determining whether $M$ can be partitioned into valid words without using diagonal movements is $\NP$-Complete. 
\end{corollary}
\begin{proof}
    The construction of \Cref{thm:strands-np-complete} is such that the Strands instance is either not solvable (not even with diagonal movements) or it is solvable without diagonal movements.
\end{proof}

Note that our main construction uses words of length 1 which are forbidden in the real version of the game. However, in our instances, this could easily be fixed by replacing each cell with a corresponding $3\times 3$ block containing the same character $9$ times. We now show that this result is more general. Indeed, it is possible to enlarge certain Strands instances so that the minimum word length is larger than some constant. Simply replacing each cell with a $3\times 3$ block containing a single character is not enough in general as adjacent cells with the same character can be problematic since we cannot force the user to cover the blocks sequentially. This can be fixed by putting an appropriate perimeter around the blocks, as shown by the following result.

\begin{theorem}
    Suppose that the Strands instance $I=(\Sigma, D, M)$ is either impossible to solve or it can be solved without diagonal movements. Suppose also that $\ell,L=O(1)$ are respectively the minimum and maximum length of a word in $D$. Then, there exists a Strands instance $I_2=(\Sigma_2, D_2, M_2)$ that is solvable if and only if $I$ is solvable and moreover: (i) if $I_2$ is solvable, then it is solvable without diagonal movements, and (ii) $|\Sigma_2|=O(|\Sigma|)$, $|D_2|=O(|D|)$ and the minimum and maximum word lengths are, respectively, $9\cdot \ell$ and $9\cdot L$.  
\end{theorem}
\begin{proof}
    For the cell in row $i$ and column $j$ we define the color $c(i,j)$ as:
    \[
        c(i,j)=
        \begin{cases}
            1 &\text{if $i$ and $j$ are even},\\
            2 &\text{if $i$ is even and $j$ is odd},\\
            3 &\text{if $i$ is odd and $j$ is even},\\
            4 &\text{if $i$ and $j$ are odd}.
        \end{cases}
    \]
    Note that $c$ is a four-coloring of the grid: no two  adjacent colors are the same. 
    
    Define $\Sigma_2=\Sigma \cup [4]$, where we assume that $\Sigma\cap [4]=\varnothing$. Suppose $M\in \Sigma^{n\times m}$. Then, the grid $M_2 \in \Sigma_2^{3n\times 3m}$ is built by replacing each $\sigma_{ij} \in M$ with a $3\times 3$ matrix. Specifically, the character $\sigma_{ij}$ is put at the center of the $3\times 3$ matrix, while its perimeter is filled with the character $c(i,j)$. 

    We now describe the new dictionary. Intuitively, we want to force a user to consume all the characters in the $3\times 3$ block before moving to a new block, while still allowing them to move in any of the four horizontal/vertical directions after having consumed the block (we will instead prevent diagonal movements). 

    For color $i\in[4]$, we define $P_i$ to be the set of colors on which we can go from $i$, that is, we let: $P_1=\{2,3\}$, $P_2=\{1,4\}$, $P_3=\{1,4\}$, and $P_4=\{2,3\}$. For an integer $k$, define $G_k$ to be the set of all the strings of the form $g \in [4]^k$ such that $g_{i+1}\in P_{g_i}$ for each $i$. Note that $|G_k|\leq 4^k$. For any $\sigma\in \Sigma$ and $i\in [4]$, let $W_i(\sigma)=i\,i\,i\,i\,i\,i\,\sigma\,i\,i \in \Sigma_2^9$. Now, for any word $w\in D$ such that $w=w_1w_2\dots w_k$, we add to $D_2$ the set of words:
    \[
        \left\{W_{g_1}(w_1)\,\, W_{g_2}(w_2)\dots \, W_{g_k}(w_k) \mid g_1g_2\dots g_k \in G_k\right\}.
    \]
    Note that $|D_2|\leq |D|\cdot 4^L = O(D)$ given that $L=O(1)$. Moreover, each $w\in D$ gets replaced by a set of words, each of length $9\cdot |w|$. 
    
    We are now left with showing that $I_2$ is solvable if and only if $I$ is solvable, and moreover, in case it is solvable, that we do not need to use diagonal movements.

    Suppose that $I$ is solvable and thus, by our assumptions, it is also solvable without using diagonal movements. Let $(i_1, j_1),\dots,(i_k,j_k)$ be a sequence of cells in the partition of the solution for $I$. Then, we can cover the corresponding $3\times 3$ blocks in $I_2$ using the word $W_{{c(i_1,j_1)}}(M_{i_1,j_1})\dots W_{{c(i_k,j_k)}}(M_{i_k,j_k})$. This is a valid word because the solution for $I$ does not move diagonally. Moreover, suppose we start a path from one of the four corners of a $3\times 3$ block for symbol $\sigma\in \Sigma$, then, by using $W_i(\sigma)$ it is possible to cover the whole block and to end up in one of two other corners. Specifically, for each side, it is possible to end up in a corner of that side without diagonal movements --- this allows us to move to an adjacent $3\times 3$ block in any of the four directions. Therefore, by transforming each sequence of cells that make up the solution of $I$, we can find a partition of $I_2$ (without using diagonal movements).
    
    On the other hand, if $I_2$ is solvable, observe that each path that constitutes the partition must necessarily be a sequence of $3\times 3$ blocks. Indeed, the coloring $c(i,j)$ is a 4-coloring of the grid, so it is not possible to move to a $3\times 3$ block without having covered the current one. Moreover, thanks to the choice of $P_1, P_2,P_3,P_4$ it is not possible to move diagonally from a $3\times 3$ block to another one. Thanks to these observations, let $\sigma_1\dots \sigma_k\in \Sigma_2^k$ be a path in the solution for $I_2$. If we remove all the symbols of $[4]$ from $\sigma_1\dots \sigma_k$, we obtain a path in $M$ that covers the cells corresponding to the $3\times 3$ blocks. Thus, we can obtain a solution to instance $I$. By the previous direction, this also implies that $I_2$ is solvable without diagonal movements. \qedhere
\end{proof}

\section{Tiles}\label{sec:tiles}
\vspace{-14mm}
\begin{figure}[H]
    \centering    \includegraphics[width=1\linewidth]{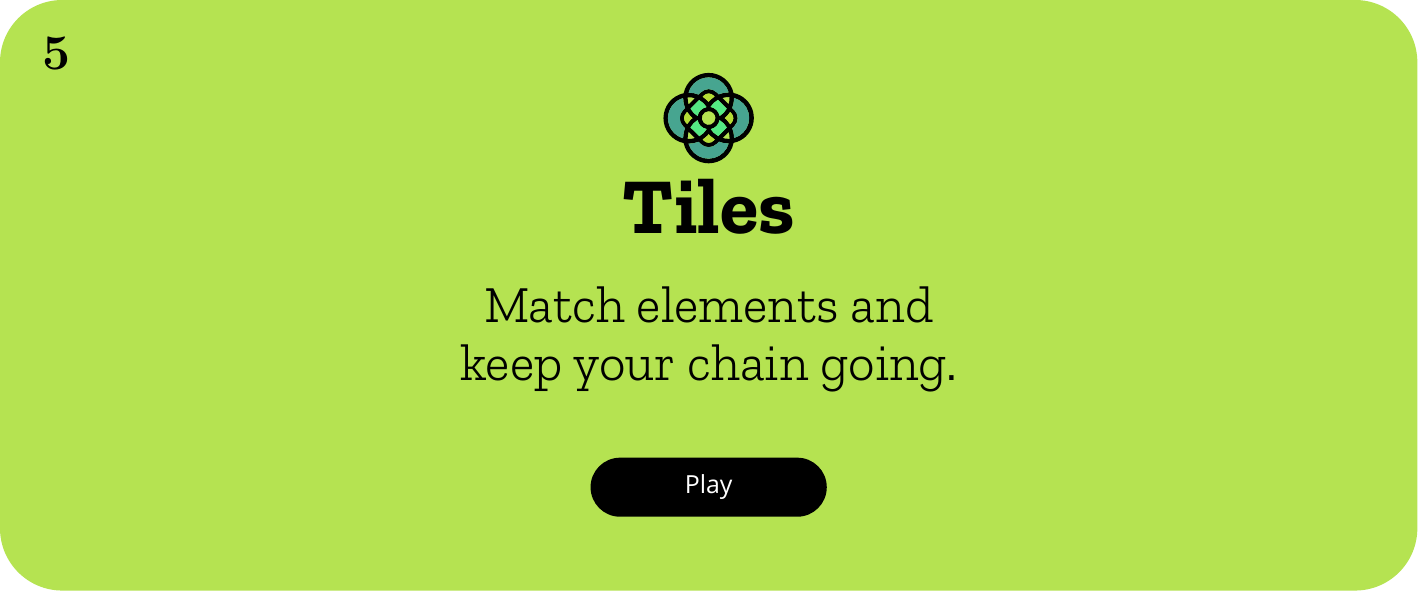}
\end{figure}

New York Times games aficionados will have undoubtedly noticed that Tiles offers a more relaxing puzzling experience than some of the other titles in the catalog. A typical instance of Tiles can be solved with little effort, aside perhaps from a little strain to the eye. This is no accident. 

In this section, we show that any (solvable) instance of Tiles can be solved by simply repeatedly making arbitrary valid moves, while favoring non-teleport whenever possible. A consequence of this result is that the problem of deciding whether an instance of Tiles is solvable is in ${P}$, and in fact one can make such a decision in time linear in the size of the puzzle.

Note that we deem a Tiles puzzle solvable if all of its features can be deleted. In general, we will show that this is true if and only if all of its features can be deleted in a single, unbroken combo. An instance is solvable with a single combo if the player can delete all the features in the instance from all of its tiles by starting at some tile and iteratively making one of two types of moves: (i) a standard move, in which the player selects a new tile that shares at least one non-deleted feature with the current tile and all the common non-deleted features between the two tiles are deleted, or (ii) a forced teleport move, in which the player, starting from a tile whose features have all been deleted, selects an arbitrary new tile.

At the end of this section, we turn to the question of whether an instance can be solved without teleport moves. We show that this can be decided in linear time for instances in which any two tiles share at most one feature.

We begin by giving a formal definition of a Tiles instance.
\begin{definition}[Tiles Instance]
    An instance of the Tiles problem $(F,\mathcal{T})$ is defined by a set $F = \{f_1, \dots , f_m\}$, the \emph{features} of the instance, and a multiset $\mathcal{T} = \{T_1, \dots , T_n\}$ of subsets of $F$, the \emph{tiles}.
\end{definition}
Intuitively, we think of every $T\in \mathcal{T}$ as the collection of all features that appear on tile $T$ (each feature can appear at most once). The size of a Tiles instance is the quantity $\sum_{T\in \mathcal{T}} |T|$. We now prove that solving a Tiles instance can be done in polynomial time. Specifically, we prove \Cref{thm:tiles-main-result}.

\TilesMainTheorem*
\begin{proof}
    The key observation to prove this theorem is that, for any fixed distinct feature in the puzzle, the parity of the number of tiles this feature appears in is an invariant: it does not change as the player makes moves. This is true since at each step in the game, a feature is either deleted from two distinct tiles, or not deleted at all.

    This immediately implies that if a feature appears in an odd number of tiles, it can never be deleted from all tiles, and hence the puzzle is unsolvable. In particular, Statement $1$ implies Statement $3$.

    On the other hand, if all features appear in an even number of tiles, one can solve the game with a single combo in a greedy fashion: by starting at an arbitrary tile and repeatedly picking a new tile that shares a feature with the current one, until no such tile exists. When this happens, the player picks an arbitrary new tile (i.e.\ they teleport). Since the parity of features is preserved, the player only teleports when the current tile has no more feature on it, and the combo continues until all features have been deleted from all tiles. Hence this shows that the instance is not just solvable, but solvable with a single combo. This shows that Statement $3$ implies Statement $2$. Since Statement $2$ trivially implies Statement $1$, this completes the proof.\qedhere

\end{proof}

We conclude this section by showing that, given an instance of Tiles in which any two tiles share at most one feature, we can decide in linear time whether this is solvable without teleport moves. Note that solving the puzzle without teleporting, implies solving the puzzle with a single combo.

Each instance of Tiles naturally induces a bipartite graph, having features on one side, and tiles on the other. Adopting this perspective will be beneficial to the result that follows.

\begin{definition}[Structure Graph of a Tiles Instance]
    Given an instance $(F,\mathcal{T})$ of Tiles, its \emph{structure graph} is the bipartite undirected graph $G=(F \cup \mathcal{T},E)$, where:
    \[
        E = \{\{T,f\}\mid T\in \mathcal{T} \text{ and }f\in T\}.
    \]
    The elements of $F$ are the \emph{feature vertices}, and the elements of $\mathcal{T}$ are the \emph{tile vertices}.
\end{definition}

Note that the structure graph of a Tiles instance is computable in linear time in the size of the instance.

\begin{definition}
    The \emph{sharing number} of a Tiles instance $(F,\mathcal{T})$ where $\mathcal{T}=\{T_1, \dots, T_n\}$ is defined as:
    \[
        \operatorname{share}(F,\mathcal{T}) = \max_{i,j\in[n], i\neq j} |T_i \cap T_j|.
    \]
\end{definition}
We now show the following.

\begin{theorem}
    An instance of Tiles with sharing number $1$ is solvable without teleport moves if and only if its structure graph admits a Eulerian trail that starts and ends on Tile vertices (potentially not distinct from one another). In particular, there is an algorithm that decides whether such an instance is solvable without teleport move in time linear in the size of the instance.
\end{theorem}
\begin{proof}
    Let $(\mathcal{F},\mathcal{T})$ be an instance of Tiles and $G =(F\cup\mathcal{T}, E)$ be its structure graph. Define:
    \[
        \mathcal{M}= \{(T_1, \dots ,T_k) \in \mathcal{T}^* \mid \forall i \in [k-1]:\;T_i\cap T_{i+1} \neq \varnothing\}
    \]
    where $\mathcal{T}^* \defeq\bigcup_{i\in\mathbb{N}}\mathcal{T}^i$, as the collection of possible sequences of moves such that any two consecutive tiles share a feature (possibly deleted). Note that the sequences of moves in $\mathcal{M}$ may contain teleport moves, as the only common feature between $T_i$ and $T_{i+1}$ might have been deleted before arriving at $i$. Now let:
    \[
        \mathcal{P}:= \{s=(T_1,f_1, \dots, f_{k-1}, T_k) \mid s\text{ is a trail in }G\text{ starting and ending in }\mathcal{T}\}.
    \]
    We define a function $\psi: \mathcal{P} \to \mathcal{M}$ as:
    \[
        \psi((T_1, f_1, \dots , f_{k-1}, T_k)) = (T_1, \dots ,T_k).
    \]
    When $\operatorname{share}(F,\mathcal{T}) = 1$, any two tiles can share at most one feature and hence this function is bijective.

    We show that $\texttt{sol}= (T_1, \dots , T_k)$ is a valid sequence of standard (non-teleport) moves that deletes all the features if and only if $\texttt{trail} = \psi^{-1}((T_1, f_1, \dots , f_{k-1}, T_k))$ is a Eulerian trail in $G$ and this will imply the statement of the theorem.

    Suppose the former is true. Consider any edge $e= \{T,f\} \in E$. Suppose $e$ is never traversed. Then, there is no value $i$ for which $T\in \{T_i, T_{i+1}\}$ and $f \in T_i\cap T_{i+1}$ but then, $f$ is never deleted from tile $T$, which contradicts the fact that $\texttt{sol}$ deletes all the features from all the tiles. On the other hand, suppose $e$ is traversed more than once by $\texttt{trail}$. Then there exists two values $i$ and $j$ with $i < j$ such that $f\in T_{i}\cap T_{i+1}$ and $f\in T_{j}\cap T_{j+1}$ and $T \in \{T_{i},T_{i+1}\}\cap\{T_j,T_{j+1}\}$. Pick the smallest two such values. $f$ is deleted from $T$ when the player goes from $T_i$ to $T_{i+1}$. This is because, since the sharing number of $(\mathcal{F},\mathcal{T})$ is $1$, the only way for a player to move from $T_i$ to $T_{i+1}$ without teleporting is if neither of them have deleted the one feature they have in common: $f$. In particular, at the $j^{th}$ time step, $T$ has already deleted feature $f$. Hence $T_j$ and $T_{j+1}$, which used to only share feature $f$, no longer share any feature, contradicting our assumption on $\texttt{sol}$.
    Hence \texttt{trail} traverses each edge exactly once, and it must be a Eulerian trail starting and ending with a tile vertex.

    For the other direction, suppose \texttt{trail} is a Eulerian trail in $G$. Then, at each time step $i$, $T_i$ and $T_{i+1}$ must share a feature $f_i$, which could not have been deleted since the edges $\{T_i,f_i\}$ and $\{T_{i+1},f_i\}$ have never been traversed before. Moreover, whenever the edge $\{T,f\}$ is traversed, the feature $f$ is deleted from $T$ and since all edges are traversed by \texttt{trail} all the features get deleted from all the tiles by \texttt{solution}, completing the proof of equivalence.

   It is well known that determining the existence of the Eulerian trail in question simply amounts to checking that: (i) the graph is connected, (ii) all vertices in $F$ have even degree, (iii) all but at most two vertices in $\mathcal{T}$ have even degree \cite{e41}. Hence it can be done in linear time.
\end{proof}

\section*{Acknowledgments}
We gratefully acknowledge the ELICSIR Foundation for its support and for fostering this collaboration. We also wish to thank the New York Times for making these truly entertaining (and hard!) games for the world to play.

Finally we wish to thank Gabriel Avram for useful suggestions on improving the images, Kyle MacMillan and Rachel B.\ Thomas for helpful advice and feedback on early versions of the paper.

Flavio Chierichetti and Alessandro Panconesi were supported in part by BiCi -- Bertinoro international Center for informatics and by a Google Focused Research Award. Flavio Chierichetti was supported in part by the PRIN project 20229BCXNW (funded by the European Union - Next Generation EU, Mission 4 Component 1 CUP B53D23012910006).

\bibliographystyle{plain}
\bibliography{bib}

\begin{thebibliography}{10}

\bibitem{addorvs15}
Aaron Adcock, Erik~D. Demaine, Martin~L. Demaine, Michael~P. O'Brien, Felix
  Reidl, Fernando~Sánchez Villaamil, and Blair~D. Sullivan.
\newblock {Zig-Zag Numberlink is NP-Complete}.
\newblock {\em Journal of Information Processing}, 23(3):239--245, 2015.

\bibitem{dd20}
Andreas Darmann and Janosch Döcker.
\newblock On a simple hard variant of {N}ot-{A}ll-{E}qual 3-{S}at.
\newblock {\em Theoretical Computer Science}, 815:147--152, 2020.

\bibitem{e41}
Leonhard Euler.
\newblock Solutio problematis ad geometriam situs pertinentis.
\newblock {\em Commentarii academiae scientiarum Petropolitanae}, pages
  128--140, 1741.

\bibitem{ghlm24}
Laurent Gourv{\`e}s, Ararat Harutyunyan, Michael Lampis, and Nikolaos
  Melissinos.
\newblock Filling crosswords is very hard.
\newblock {\em Theoretical Computer Science}, 982:114275, 2024.

\bibitem{H21}
Josie Harvey.
\newblock {New York Times Accidentally Publishes Mock Article About Watermelons
  On Mars}.
\newblock HuffPost UK.
  \url{https://www.huffingtonpost.co.uk/entry/watermelons-mars-new-york-times-article_uk_60c0759be4b0e6bab7a20463}.
\newblock Trends Reporter, Accessed on 2025-09-12.

\bibitem{k09}
Richard~M Karp.
\newblock Reducibility among combinatorial problems.
\newblock In {\em 50 Years of Integer Programming 1958-2008: from the Early
  Years to the State-of-the-Art}, pages 219--241. Springer, 2009.

\bibitem{Klein2023NYTimesGames}
Charlotte Klein.
\newblock {Inside The New York Times’ Big Bet on Games}.
\newblock Vanity Fair.
  \url{https://www.vanityfair.com/news/inside-the-new-york-times-big-bet-on-games},
  2023.
\newblock Accessed: 2025-06-26.

\bibitem{Knight2024AtlanticStrands}
Jonathan Knight.
\newblock The {T}imes’ {N}ew {W}ord‑{S}earch {G}ame is {G}enius.
\newblock The Atlantic.
  \url{https://www.theatlantic.com/technology/archive/2024/03/strands-word-search-game/677656/},
  2024.
\newblock Accessed: 2025-06-26.

\bibitem{Levine2024Strands}
Elie Levine.
\newblock Putting a new twist on a classic puzzle.
\newblock The New York Times.
  \url{https://www.nytimes.com/2024/03/04/crosswords/strands-word-search-game.html},
  2024.
\newblock Accessed: 2025-06-26.

\bibitem{LS22wordle}
Daniel Lokshtanov and Bernardo Subercaseaux.
\newblock Wordle is {NP}-hard.
\newblock In {\em Fun - 11th International Conference on Fun with Algorithms},
  2022.

\bibitem{m23}
Aisha Majid.
\newblock 100k club: Exclusive ranking of world’s top paywalled news
  publishers.
\newblock Press Gazette
  \url{https://pressgazette.co.uk/paywalls/100k-club-exclusive-ranking-of-worlds-top-paywalled-news-publishers/},
  2023.
\newblock Accessed: 2025-09-12.

\bibitem{WinNT}
ValueAct~Capital Management.
\newblock Notice of exempt solicitation submitted pursuant to rule 14a-6(g).
\newblock
  \url{https://www.sec.gov/Archives/edgar/data/1351069/000119312524062156/d809296dpx14a6g.htm}.
\newblock Accessed: 2025-06-26.

\bibitem{mr08}
Wolfgang Mulzer and G\"{u}nter Rote.
\newblock Minimum-weight triangulation is {NP}-hard.
\newblock {\em J. ACM}, 55(2), 2008.

\bibitem{Peters2025NYTScrabble}
Jay Peters.
\newblock {The New York Times} is starting to test its take on scrabble.
\newblock The Verge.
  \url{https://www.theverge.com/games/682063/the-new-york-times-nyt-games-crossplay-scrabble-pips},
  June 2025.
\newblock Accessed: 2025-09-11.

\bibitem{r22}
Will Rosenbaum.
\newblock Finding a winning strategy for wordle is {NP}-complete.
\newblock {\em arXiv}, 2204.04104, 2022.

\bibitem{Shepard92}
Richard~F. Shepard.
\newblock {B}ambi {I}s a {S}tag, and {T}ubas {D}on't {G}o '{P}ah{P}ah': {T}he
  {I}ns and {O}uts of '{A}cross'.
\newblock The {N}ew {Y}ork {T}imes {M}agazine
  \url{https://www.nytimes.com/1992/02/16/magazine/bambi-is-a-stag-and-tubas-dont-go-pahpah-the-ins-and-outs-of-across.html},
  February 1992.
\newblock Accessed: 2025-07-16.

\bibitem{Switzer2024StrandsGamer}
Eric Switzer.
\newblock {N}ew {Y}ork {T}imes {H}ot {N}ew {W}ord {G}ame {J}ust {M}ade
  {E}veryone {F}reak {O}ut.
\newblock TheGamer.
  \url{https://www.thegamer.com/new-york-times-hot-new-word-game-just-made-everyone-freak-out/},
  2024.
\newblock Accessed: 2025-06-26.

\bibitem{nyt_pulitzers}
{The New York Times Company}.
\newblock Pulitzer prize winners.
\newblock
  \url{https://www.nytco.com/award-collection/2019-pulitzer-prize-winners/}.
\newblock Accessed: 2025-07-17.

\bibitem{nyt2024intlsubs}
{The New York Times Company}.
\newblock {The New York Times Passes 2m International Digital Subscribers}.
\newblock
  \url{https://www.nytco.com/press/the-new-york-times-passes-2m-international-digital-subscribers/},
  June 2024.

\bibitem{nyt2025pips}
{The New York Times Company}.
\newblock Introducing pips, a new logic game from new york times games.
\newblock
  \url{https://www.nytco.com/press/introducing-pips-a-new-logic-game-from-new-york-times-games/},
  August 2025.
\newblock Press release.

\bibitem{W06}
Lilian Whiting.
\newblock {There Is Life on the Planet Mars; Prof. Percival Lowell Recognized
  as the Greatest Authority on the Subject, Declares There Can Be No Doubt That
  Living Beings Inhabit Our Neighbor World}.
\newblock {\em New York Times}, 1906.

\bibitem{ys03}
Takayuki Yato and Takahiro Seta.
\newblock Complexity and completeness of finding another solution and its
  application to puzzles.
\newblock {\em IEICE transactions on fundamentals of electronics,
  communications and computer sciences}, 86(5):1052--1060, 2003.

\end{thebibliography}

\end{document}